\documentclass[a4paper,10pt,twoside,reqno,nonamelimits]{article}
\usepackage{a4wide}
\usepackage[T1]{fontenc}
\usepackage[colorlinks=true,urlcolor=blue,linkcolor=blue,citecolor=blue]{hyperref}
\usepackage{newcent}         

\usepackage{amsmath}
\usepackage{amssymb}
\usepackage{amsfonts}
\usepackage{amscd}
\usepackage{amsthm}
\usepackage{amsxtra}
\usepackage{wasysym}
\usepackage{mathrsfs}
\usepackage{nicefrac}
\usepackage{graphicx}
\usepackage{xcolor}
\usepackage{tikz}
\usepackage[inline]{enumitem}
\setlist[enumerate]{leftmargin=2em,itemindent=0em, labelindent=0pt,labelwidth=1.5em,labelsep=.5em, align=left, noitemsep}
\newlist{condenum}{enumerate}{1}
\setlist[condenum]{ leftmargin=4em,itemindent=0em, labelindent=0pt,labelwidth=1.5em,labelsep=.5em, align=left, noitemsep}
\newlist{txtenum}{enumerate}{1}
\setlist[txtenum]{leftmargin=0em,itemindent=1.5em, labelindent=0pt,labelwidth=1em,labelsep=.5em, align=left}

\usepackage[boxruled,noend]{algorithm2e}\DontPrintSemicolon 
\SetKwInOut{Fixed}{Fixed}
\SetKwInOut{Input}{Input}
\SetKwInOut{Output}{Output}

\SetKwFor{SeparationRoutine}{Separation Routine}{}{end}
\SetKwRepeat{Repeat}{Repeat}{Until}
\SetKwIF{If}{ElseIf}{Else}{If}{then}{Else if}{Else}{Endif}
\SetKw{Return}{Return}

\usepackage{titlesec}
\titleformat{\subsubsection}[runin]{\normalfont\normalsize\bfseries}{\thesubsubsection.\ }{0pt}{}[.]

\renewcommand{\thesubsubsection}{\thesubsection.\alph{subsubsection}}

\newcommand{\mypar}{\par\medskip\noindent}

\usepackage{datetime}

\theoremstyle{plain}
\newtheorem{theorem}{Theorem}
\newtheorem*{theorem*}{Theorem}

\newtheorem{proposition}[theorem]{Proposition}
\newtheorem*{proposition*}{Proposition}

\newtheorem*{corollary*}{Corollary}

\newtheorem{lemma}[theorem]{Lemma}
\newtheorem*{lemma*}{Lemma}

\newtheorem*{observation*}{Observation}

\newtheorem*{remark*}{Remark}

\newtheorem*{remarks*}{Remarks}

\theoremstyle{definition}

\newtheorem*{definition*}{Definition}

\newtheorem*{notation*}{Notation}

\newtheorem*{example*}{Example}

\newtheorem*{exercise*}{Exercise}

\theoremstyle{remark}

\newtheorem*{claim*}{Claim}

\newtheorem*{conjecture*}{Conjecture}

\newtheorem*{question*}{Question}

\newtheorem*{questions*}{Questions}

\newtheorem*{problem*}{Problem}

\newtheorem*{problems*}{Problems}

\newtheorem*{openproblem*}{Open Problem}

\newcommand{\subclass}[1]{}

\newcommand{\enumTi}[1]{\renewcommand{\theenumi}{#1}}

\newcommand{\alphenumi}{\enumTi{\alph{enumi}}}

\newcommand{\romenumi}{\enumTi{\roman{enumi}}}

\newlength{\hspaceforlengthglumpf}

\renewcommand{\em}{\sl}

\newcommand{\comment}[1]{\text{\footnotesize[#1]}}

\DeclareMathOperator{\tr}{tr}

\DeclareMathOperator{\Supp}{Supp}
\DeclareMathOperator{\Spec}{Spec}

\newcommand{\lt}{\left}
\newcommand{\rt}{\right}

\newcommand{\T}{{\mspace{-1mu}\scriptscriptstyle\top\mspace{-1mu}}}

\newcommand{\sstack}[1]{{\substack{#1}}}

\DeclareMathOperator{\pr}{pr}

\newcommand{\pii}{{\pi i}}

\newcommand{\nfrac}[2]{{\nicefrac{#1}{#2}}}

\newcommand{\CC}{\mathbb{C}}

\newcommand{\NN}{\mathbb{N}}

\newcommand{\QQ}{\mathbb{Q}}
\newcommand{\RR}{\mathbb{R}}

\newcommand{\ZZ}{\mathbb{Z}}

\DeclareMathOperator*{\Exp}{\mathbb{E}}

\newcommand{\eps}{\varepsilon}

\newcommand{\ceil}[1]{\lceil{#1}\rceil}

\newcommand{\abs}[1]{{\lvert{#1}\rvert}}                        \newcommand{\Nm}[1]{{\lVert{#1}\rVert}}

\newcommand{\absB}[1]{{\Bigl\lvert{#1}\Bigr\rvert}}

\newcommand{\bra}[1]{{\langle #1 \rvert}}
\newcommand{\ket}[1]{{\lvert  #1 \rangle}}

\newlength{\algotabbingwidth}
\DeclareMathOperator{\cost}{cost}
\newcommand{\Fun}{\mathscr{K}}            
\newcommand{\place}{{\textrm{\tiny$\,\square\,$}}}
\DeclareMathOperator{\Diff}{Diff}
\newcommand{\Mu}{\mathrm{M}}

\DeclareMathOperator{\cvxcone}{cvxcone}

\theoremstyle{plain}
\newtheorem{computationalmethod}[theorem]{Computational Method}

\begin{document}
\title{Optimality of Finite-Support Parameter Shift Rules for Derivatives of Variational Quantum
  Circuits}%
\author{Dirk Oliver Theis %
  \\[1ex]
  \small Theoretical Computer Science, University of Tartu, Estonia\\
  \small \texttt{dotheis@ut.ee}%
}
\date{Tue 22 Mar 2022 11:36:19 AM CET
  \\
  Compiled: \today, \currenttime}
\maketitle

\begin{abstract}
  Variational (or, parameterized) quantum circuits are quantum circuits that contain real-number parameters, that need to be
  optimized/``trained'' in order to achieve the desired quantum-computational effect.  For that training, analytic derivatives
  (as opposed to numerical derivation) are useful.  Parameter shift rules have received attention as a way to obtain analytic
  derivatives, via statistical estimators.

  In this paper, using Fourier Analysis, we characterize the set of all shift rules that realize a given fixed partial
  derivative of a multi-parameter VQC.  Then, using Convex Optimization, we show how the search for the shift rule with smallest
  standard deviation leads to a primal-dual pair of convex optimization problems.

  We study these optimization problems theoretically, prove a strong duality theorem, and use it to establish optimal dual
  solutions for several families of VQCs.  This also proves optimality for some known shift rules and answers the odd open
  question.

  As a byproduct, we demonstrate how optimal shift rules can be found efficiently computationally, and how the known optimal
  dual solutions help with that.

  \par\medskip%
\end{abstract}

\setcounter{tocdepth}{3}
\tableofcontents

\section{Introduction}\label{sec:intro}
Near-term quantum algorithms meant to be run on pre-fault-tolerant quantum computing devices often employ the concept of
containing parameters, which are being optimized/``trained'' to maximize the utility of the quantum algorithm for the particular
application.  The concept is known, depending on the application area, as Variational Quantum Program, Parameterized Quantum
Circuit, Quantum Neural Network, etc (see \cite{Cerezo-Arrasmith-Babbush-Benjamin-etal:VQAs:2021} and the references therein).

Typical examples include the various types of Quantum Neural Networks for quantum machine learning applications (e.g.,
\cite{Mitarai-Negoro-Kitagawa-Fujii:qcirclearn:2018,Benedetti-Lloyd-Sack-Fiorentini:PQCs:2019}), Variational Quantum
Eigensolvers for electronic structure computations (e.g., \cite{Yuan-Suguru-Qi-Ying-Benjamin:theory-VQS:2019}), and, to varying
degrees, Quantum Approximate Optimization Algorithms~\cite{Farhi-Goldstone-Gutmann:QAOA:2014}.

One way of training the parameters of a variational or parameterized quantum circuit deploys some type of smooth optimization
algorithms such as (stochastic) gradient descent or second order methods or whatnot, in other words, it requires derivatives, be
they first order (e.g.,
\cite{Mitarai-Negoro-Kitagawa-Fujii:qcirclearn:2018,Schuld-Bergholm-Gogolin-Izaac-Killoran:eval-anal-grad:2019}), or higher
order (e.g., \cite{Mari-Bromley-Killoran:higher-order:2021, Stokes-Izaac-Killoran-Carleo:qnatgrad:2020}).  While, of course, the
standard formulas of numerical analysis yield approximations for derivatives, the interest is in so-call \textsl{analytic}
derivatives, which don't approximate (with some $\eps>0$ that has to be chosen carefully) the derivatives, but estimate the
exact derivative.  Tickling these analytic derivatives, of any order, out of the quantum circuit is the subject matter of this
paper.

Interestingly, the interest in efficient methods for estimating derivatives seems to have become more and more interesting,
e.g., \cite{Crooks:param-shift:2019, Hubregtsen-Wilde-Qasim-Eisert:grad-rules:2021, Banchi-Crooks:stoch-param-shift:2021,
  Kottmann-Anand-AspuruGuzik:diffUCC:2021, Izmaylov-Lang-Yen:algebraic-param-shift:2021,
  Wierichs-Izaac-Wang-Lin:gen-shift:2021,Kyriienko-Elfving:genqcdiff:2021}.

\mypar%
In this paper, we consider parameterized quantum circuits which, with parameters set to $x\in\RR^d$, apply a quantum operation
$\mathcal{E}(x)$ to an initial quantum state~$\varrho_0$, and in the resulting state, a fixed observable~$\mu$ is measured.  The
expectation value $f(x) := \tr( \mu \mathcal{E}(x)(\varrho_0))$ is the quantity of which one wishes to find the derivatives.  In
other words, we are interested in the (partial) derivatives of the \textit{expectation value function}
\begin{equation*}
  f\colon
  \RR^d \to \RR
  \colon
  x \mapsto \tr( \mu  \mathcal{E}(x)(\varrho_0)  ).
\end{equation*}

Moreover, we consider the situation where each of the parameters enter the parameterized quantum circuit via one-parameter
subgroups of unitary operators, i.e., mappings $U(\place)$ with $U(0)$ equal to the identity and for all $s,t\in\RR$, we have
$U(s+t) = U(s)U(t)$.  Each parameter of the PQC could appear in one or several places of the quantum circuit, with the same or
different one-parameter subgroups.  One-parameter subgroups are, of course, of the form $U(\place) = e^{2\pii \place H}$ for a
Hermitian operator~$H$.  (Note the author's idiosyncratic scaling factor $2\pi$, that matches the standard Fourier transform in
Harmonic Analysis.)

This is the setting that seems to be considered in the majority of derivative estimation research, although other settings
exist, e.g., the ``stochastic'' parameter shift rules of~\cite{Banchi-Crooks:stoch-param-shift:2021}, which applies to
parameterized unitaries of the form $t\mapsto e^{i(tA+B)}$.  Our results don't cover these types of parameterized unitaries.

To our knowledge, analytic derivatives of Parameterized Quantum Circuits have been introduced into near-term quantum
programming in~\cite{Mitarai-Negoro-Kitagawa-Fujii:qcirclearn:2018}.

Soon after~\cite{Mitarai-Negoro-Kitagawa-Fujii:qcirclearn:2018}, the trivial observation was made that, in full generality, the
expectation value functions of parameterized quantum circuits into which the parameters enter through one-parameter subgroups
can be characterized in terms of their Fourier properties
\cite{GilVidal-Theis:CalcPQC:2018,Crooks:param-shift:2019,gilvidal-theis:inpu-redu:2019}.  This fact can be exploited directly
for the optimization/training of the PQC \cite{GilVidal-Theis:CalcPQC:2018,Ostaszewski-Grant-Benedetti:rotosolve:2021}, but it
also has bearing on shift rules for derivatives, most notably in
\cite{Wierichs-Izaac-Wang-Lin:gen-shift:2021,Kyriienko-Elfving:genqcdiff:2021}, the main references of the present paper.

\mypar%
Indeed, the present paper applies Harmonic Analysis and Convex Optimization to prove the optimality, in terms of the standard
deviation of the resulting estimator, of some parameter shift rules.  Unfortunately, the shift rules for which we prove
optimality are \textsl{mostly} known from~\cite{Wierichs-Izaac-Wang-Lin:gen-shift:2021}.  However, our paper adds the certainty
that no better shift rules (in terms of stddev) exist.

\paragraph{Basic approach.} %
We find that things get simpler if we abstract from the quantum-circuit origin of the functions and work with their Fourier
analytic properties: Expectation value functions are real valued, and for a parameterized quantum circuit with~$d$ parameters
(it's OK to count only parameters participating in the partial derivative), there exists a finite set $\Xi\subset\RR^d$ which
contains the Fourier spectrum of the expectation value function.  As the function is real valued, the set~$\Xi$ is
\textit{symmetric,} meaning:
\begin{equation}\label{def:Xi-symmetric}
  \Xi = -\Xi \text{ where $-\Xi := \{-\xi\mid\xi\in\Xi\}$.}
\end{equation}
We refer to finite, symmetric sets~$\Xi\subset\RR^d$ as \textit{frequency sets} of PQCs.

While every such set can occur as the support of the expectation value function of a Parameterized Quantum Circuit (see
Proposition~\ref{prop:overview:FourierComputability} below) we acknowledge that, for $d>1$, only a small subclass of frequency
sets are truely of interest, mostly because a (small) set~$\Xi$ containing the Fourier spectrum has to be somehow read off from
the PQC before shift rules can be constructed and applied.  See \S\ref{ssec:overview:fn-spaces} below for (mostly known)
examples of frequency sets of expectation value functions of PQCs which can be found by inspecting the parameterized quantum
circuit.

Before we summarize some of the contributions of this paper, we draw the reader's attention to Table~\ref{tab:notations} which
lists the notations frequently used in this paper, for easy reference.

\begin{table}[th]\sffamily\small\centering
  \begin{tabular}{lp{.75\linewidth}}
    \bfseries Notation       & \bfseries Description
    \\[1ex]
    \hline
    \\
    $\NN := \{1,2,3,\dots\}$ &
    \\[1ex]
    $d\in\NN$                & Number of parameters of the function / parameterized quantum circuit.
    \\[1ex]
    $\alpha\in\NN^d$         & The order of the multi-variable partial derivative.  The reason why we don't allow $\alpha_j=0$ for some $j$'s is just convenience: We can restrict to the parameters that contribute in the partial derivative.
    \\[1ex]
    $\partial^\alpha = \prod_{j=1}^d \partial_j^{\alpha_j}$
    {}                       & Partial derivative in the math sense, see \href{https://en.wikipedia.org/wiki/Partial_derivative}{Wikipedia}.
    \\[1ex]
    $\abs{\alpha}$           & For $\alpha\in\NN^d$: \quad $\abs{\alpha}:=\sum_{j=1}^d \alpha_j$.
    \\[1ex]
    $\Xi \subset\RR^d$       & A finite set (of frequencies), elements are typically $\xi\in\Xi$
    \\[1ex]
    $\Xi_+$                  & ``Positive'' part of $\Xi$: Arbitrary but fixed set satisfying the relation $\Xi = \{0\}\cup \Xi_+ \cup (-\Xi_+)$, where the union is disjoint.  In case $d=1$, we think of $\Xi\_\subset\lt]0,\infty\rt[$
    \\[1ex]
    $A\subset \RR^d$         & A (finite) set of points, elements are typically $a\in A$.  The reader will understand the reason for writing $f(a)$ for (nice)
                               functions $f\colon\RR^d\to\RR$.  We appreciate that writing both $f(x)$ and $f(a)$, for the same function~$f$, must be terribly
                               confusing for physicists.
    \\[1ex]
    $\place$                 & Placeholder in expressions to simplify notations for functions.  The placeholder is where you plug the variable, which as no name yet.  E.g., as already used above,
                               $e^{2\pii\place H}$ stands for $v\mapsto e^{2\pii v H}$.  (There's no preference for a particular letter; to emphasize that, with~$v$ we deliberately made a weird choice.)
    \\[1ex]
    $u\bullet v$             & $:= \sum_j u_j v_j$.
    \\[1ex]
    $\Fun_\Xi$               & Real vector space of bounded real-valued functions whose Fourier spectrum is contained in the
                               finite set~$\Xi$; i.e., functions of the form $\sum_{\xi\in\Xi} b_\xi e^{2\pii\xi\bullet x}$
                               with $b_{-\xi}=b_\xi^*$ for all $\xi\in\Xi$.
  \end{tabular}
  \caption{\label{tab:notations}\textbf{Overview of notations.}  %
    \small%
    Habits of notation in Physics and Math / Theoretical Computer Science are very different, which is why we list the notations
    which are used a throughout this paper, particularly the arcane ones.}
\end{table}

\paragraph{The contributions of this paper.} %
(1) Using elementary Harmonic Analysis, for any $d\in\NN$ and any given finite symmetric set $\Xi\subset\RR^d$, and for fixed
$\alpha\in\NN^d$, we characterize the set of all shift rules which realize the partial derivative\footnote{In the math sense,
  \href{https://en.wikipedia.org/wiki/Partial_derivative}{en.wikipedia.org/wiki/Partial\_derivative}.}
$\partial^\alpha := \prod_{j=1}^d \partial_j^{\alpha_j}$ for every real-valued function with Fourier spectrum contained
in~$\Xi$.  We refer to these shift rules as ``feasible''.  As we explain in~\S\ref{ssec:overview:sys-eqns}, at the heart of our
characterization is a finite, underdetermined linear system of equations, not unlike, but more general than those that occur in
\cite{Wierichs-Izaac-Wang-Lin:gen-shift:2021, Kyriienko-Elfving:genqcdiff:2021} where the systems are used to search for shift
rules.

\mypar%
Our more general characterization is necessary because it allows us to (2) realize the search for the optimal shift rule (among
all possible shift rules) as a the problem of minimizing a convex function over that set
(\S~\ref{ssec:overview:primal+dual+weakduality}).

\mypar%
The word ``optimal'', of course, requires the concept of a ``cost'' of a shift rule.  We use a cost measure closely related to
(the same, really) one proposed in~\cite{Wierichs-Izaac-Wang-Lin:gen-shift:2021} and called the ``number of shots'' there.  It
can be understood as (2a) an (asymptotic) standard deviation of an estimator in the worst case over all parameterized quantum
circuits; or, (2b) for the Functional/Numerical Analyst, as the operator norm of the shift rule when viewed as a linear operator
between the obvious normed spaces.  As usual, the latter perspective allows us to view the ``cost'' as the worst-case bound on
how noise in the input function (e.g., sampled estimate of the expectation value function) is blown up in the output function
(e.g., estimate of the derivative of the expectation value function).

As a nice coincidence, our cost measure can also be understood (2c) as a heuristic to minimize the quantity that
in~\cite{Wierichs-Izaac-Wang-Lin:gen-shift:2021} is called the ``number of unique circuits'', i.e., the number of ``shift
vectors'' used --- we refer to it as the size of the support of the shift rule.

See \S\ref{ssec:overview:cost} for the formal definition of our cost function, and for discussions related to (2a,2b,2c).

We emphasize that in this paper the phrase ``cost of a shift rule'' refers to one fixed mathematical concept (with three
different justifications) --- we do not switch between different resource requirement concepts.  However, we also emphasize that
much of the technical toolkit we develop carries through to other cost functions: For other ``linearizable'' cost functions (our
cost function has that property, see \S\ref{ssec:overview:computation:primal}) most of our tools still work; for general convex
cost functions the duality concept (see below) will suffer; for cost functions that require on/off decisions (such as for the
provably-optimal minimization of the size of the support), we exit the land of comforting convexity, and enter the realm of NP
hardness.

\mypar%
As we are interested only in ``finite'' shift rules, i.e., where the size of the support is finite, the resulting optimization
problem has a finitness condition which, ostensibly, falls out of the scope of standard convex optimization.  Nonetheless, (3)
we establish a perfectly standard ``dual'' convex optimization problem
(\S\ref{ssec:overview:primal+dual+weakduality:weakduality}), and prove two strong duality theorems: A pair of primal and dual
solutions are both optimal, if, and only if, (4a) so-called complementary slackness holds, (4b) the objective function values
are the same; see \S\ref{ssec:overview:strong-duality}.

\mypar%
(5) \textsl{En passant} in~\S\ref{ssec:overview:computation}, we outline how, in the cases where optimal shift rules cannot be
found on paper, standard computational convex optimization techniques when applied to the dual convex program yield optimal
solutions to it, and how, from these dual-optimal solutions, optimal finite feasible shift rules can be obtained
computationally.  The algorithms are straightforward consequences of our work on the duality theory (2,3,4).

\begin{computationalmethod}
  Let $d\in\NN$ be a fixed parameter.

  There exists an efficient algorithm for the following task:
  \begin{quote}
    Given as input a finite symmetric set $\Xi\subset\QQ^d$ and $\alpha\in\NN^d$, and a precision parameter $\eps>0$, return a
    near-optimal near-feasible shift rule (where ``near'' is controlled through the precision parameter).
  \end{quote}
  Moreover, heuristically, in cases where there are several optimal shift rules, one can hope for a shift rule with small
  support.
\end{computationalmethod}

We also have a simpler method for the cases when optimal dual solutions can be obtained analytically (which happens often, see
below): There, a single Linear Program has to be solved (\S\ref{ssec:overview:computation:primal-w-knowledge}).

The computational cost of the methods depends on $d$ and~$1/\eps$.

The precision parameter is always necessary, firstly simply because non-rational numbers are involved due to the trigonometric
functions that occur, and secondly because the theoretical convergence results require them.  While everything else in this
paper is (meant to be) fully rigorous, we will not prove that our algorithm(s) run in polynomial time (for\footnote{%
  The reason why~$d$ has to be fixed is that the methods involve minimizing a function with Fourier spectrum contained
  in~$\Xi\subset\QQ^d$, which can be done efficiently only for fixed~$d$.  } %
fixed~$d$): They are typical, standard approaches in mathematical optimization that are widely used both in practice
(cutting-plane/column generation methods) and in computational complexity theory (proofs of poly-time running times based on the
Ellipsoid method).

Instead, we invite the reader to check out the supplementary source code written in the mathematical computation programming
language Julia\footnote{\texttt{\href{https://julialang.org/}{julialang.org}}} demonstrating the methods in the form of
Pluto\footnote{\texttt{\href{https://plutojl.org}{plutojl.org}}} notebooks; see Appendix~\ref{apx:source-code}.  We point out
that the source code is only for demonstration of the algorithms, and will run into numerical instabilities when the frequency
sets get a little bit nasty.

\mypar%
In \cite{Wierichs-Izaac-Wang-Lin:gen-shift:2021, Kyriienko-Elfving:genqcdiff:2021}, searches for shift rules based on more
restricted (linear) system characterizations are performed manually.  But to our knowledge, there seems to be no method to
efficiently computationally find cost-wise (near-) optimal shift rules for non-trivial frequency sets (even in dimension $d=1$),
and they don't seem to be known analytically, either.

\mypar%
(6) Finally, we use the duality theorems to prove optimal shift rules for some situations arising from parameterized quantum
circuits, see \S\ref{ssec:overview:the-list}.

For that, we need two distinguish two types of frequency sets --- ones that generate lattices in~$\RR^d$, which we call
\textit{lattice generating}, and ones that generate dense subgroups of sub-vector spaces of $\RR^d$.  We'll elaborate on that
distinction (known to every mathematician) below (\S\ref{ssec:overview:the-list}).  For now we just mention that (a) a function
is periodic if, and only if, its Fourier spectrum is lattice generating (Appendix~\ref{apx:miscmath:lattices-and-periodicity});
(b) a set is lattice generating if, and only if, after a change of basis (``re-scaling''), it is integral (mathematical
folklore); (c) for example, for $d=1$, if every element of the frequency set is a rational multiple of a fixed real constant,
then the frequency set is lattice generating.

(6a) For any $d,\alpha, \Xi$, we prove that a finite feasible \textsl{optimal} shift rule exists, and Wierichs et al.'s quantity
``number of unique circuits'' (what we call the size of the support of the shift rule) is at most~$\abs{\Xi}$
(\S\ref{ssec:overview:ex-opt}).  (Examples are known that show that the bound is tight.)

(6b) For $d=1$, any frequency set $\Xi\subset\RR$, and any\footnote{In this paper, we give the complete proof only for
  $\alpha\le 4$ and we refer to a forthcoming paper for the proof of a technical lemma for $\alpha>4$.} $\alpha\in\NN$, we
quantify the cost of optimal shift rules: It is $(2\pi\max\Xi)^\alpha$.  In the simplest case, if the frequency set has the form
(for a positive integer~$m$ and a positive real number~$\xi_1$)
\begin{equation}\label{eq:simplest-Xi}
  \Xi = \{ k\xi_1 \mid k\in\ZZ, \abs{k} \le m\},
\end{equation}
we prove the optimality of the shift rule that was given in~\cite{Wierichs-Izaac-Wang-Lin:gen-shift:2021} for that case.  But,
for all other lattice generating~$\Xi$, we prove that the corresponding shift rule is still optimal in terms of its
\textsl{cost!}  In other words, it is impossible for a shift rule to exploit ``gaps'' in the set of frequencies to obtain lower
cost.  This closes an issue that was left open in~\cite{Wierichs-Izaac-Wang-Lin:gen-shift:2021}.  However, regarding the size of
the support, in view of (6a), it is clear that for sparse frequency sets, there are optimal shift rules with different
support sizes.  As outlined above, our \textsl{computational} method can be heuristically hoped to find a shift rule with small
support.

Particularly for $d=1$, we show that an optimal shift rule exists whose support size is at most $\abs{\Xi}-(\alpha-1)\%2$.  This
means that there are optimal shift rules among those described in general forms in
in~\cite{Wierichs-Izaac-Wang-Lin:gen-shift:2021,Kyriienko-Elfving:genqcdiff:2021}.  We point out that our simple formula for the
cost of optimal finite-support shift rules holds in full generality, i.e., also for non-lattice-generating~$\Xi$.

(6c) For $d\ge 2$ our results are not comprehensive.  We quantify the cost of optimal shift rules only for frequency sets
satisfying a restrictive technical condition (``pointy'').  The condition includes the case when~$\Xi$ is a product set, which
is the only one that if considered in~\cite{Wierichs-Izaac-Wang-Lin:gen-shift:2021} --- it is questionable, though, whether more
complicated frequency sets can be identified \textsl{a priori,} by inspecting the parameterized quantum circuit or performing a
small number of expectation-value estimations that probably shouldn't depend on any other parameters.

Note that in this paper, we don't consider savings that can be gained by computing several partial derivatives, e.g., a complete
Hessian, as is done in~\cite{Wierichs-Izaac-Wang-Lin:gen-shift:2021}.

The results (6b) and (6c), described in detail in \S\ref{ssec:overview:the-list:d=1} and~\S\ref{ssec:overview:the-list:d-ge-2}
resp., are obtained by establishing a dual-optimal solution --- we consider these dual-optimal solutions to be a major
contribution of this paper.  Once we have dual-optimal solutions, using our other theorems, the (existence and) cost of optimal
shift rules can easily be derived, and one of our computational methods (see above) demonstrates a possible simple way to
computationally find a heuristically-sparse optimal shift rule even for rugged frequency sets, based on the dual-optimal
analytic solution.

\paragraph{About non lattice generating frequency sets.}
It is prudent to point out a danger of misunderstanding~\cite{Wierichs-Izaac-Wang-Lin:gen-shift:2021}: There, the innocent
reader might think of a dichotomy between an equi-spaced Fourier spectrum $\{k\xi_1 \mid k=-K,\dots,K\}$ on the one hand, and
non-periodicity on the other.  While functions with equi-spaced Fourier spectrum are periodic, the converse is, of course, not
true.  Indeed, the periodicity of the function is equivalent to the Fourier spectrum generating a lattice.

Wierichs et al.~(Appendix~B.1) talk about frequencies which cannot be scaled to be integral, and use the phrase \textsl{not
  periodic} in that context, in the case $d=1$.  It is not clear to the author of the current paper what they mean.

Note also that a frequency set that is not lattice generating cannot be represented in a computer: A representation of the
vectors/points in any basis will always require coefficients that are not rational.


Why bother with non-lattice-generating Fourier spectra at all?  Clearly, from a practical point of view, all frequency sets are
lattice generating, and the treatment of non-lattice-generating frequency sets brings some difficulty.  (For example, it is not
clear whether an optimal feasible shift rule exists, i.e., whether the infimum of the costs is not attained.)

The reason is that lattice generating frequency sets can sometimes blow up the complexity.  Consider, say, 10 31-bit unsigned
fixed-point numbers $\xi_1,\dots,\xi_{10} \in \lt]0,2\rt[$ chosen uniformly at random, and consider a corresponding (random)
real-valued function that has, as its Fourier spectrum, the set $\Xi := \{0\} \cup \{ \pm\xi_j \mid j=1,\dots,10 \}$.  Clearly,
$\Xi$ is lattice generating, but the probability that its period is smaller than $2^{29}$ is at roughly one in a million.
Hence, treating this frequency set as lattice generating, while correct, does not make the problem of finding a finite optimal
shift rule feasible.  Considering non-lattice-generating frequency sets, on the other hand, and trading elusive optimality (on
account of non-representability in a computer memory) for near-optimality (off by at most an arbitrary given $\eps>0$),
immediately leads to tools such as simultaneous Diophantine approximation (and so to continued fractions, LLL algorithm, etc).

While that offer sufficient motivation, we believe, for not excluding non-lattice-generating frequency sets, in the present
paper, we do not pursue the algorithmic approaches that they inspire.

\paragraph{Outline of the paper.} %
In the next section, describe the function spaces we are working with.
In Section~\ref{ssec:overview:def-feasible-shift-rule}, we define (feasible) shit rules rigorously.
In Section~\ref{ssec:overview:sys-eqns}, we give the characterization of the set of feasible shift rules.
In Section~\ref{ssec:overview:cost}, we define the cost of a shift rule, and show the equivalence of~3 possible definitions.
In Section~\ref{ssec:overview:primal+dual+weakduality}, we define the ``primal'' and ``dual'' optimization problems for shift
rules and prove the Weak Duality Theorem.
In Section~\ref{ssec:overview:computation}, we interject a discussion on various aspects of the optimization problems that can
be computed in practice (and, indeed, in theory).
In Section~\ref{ssec:overview:strong-duality}, we state the Strong Duality and Complementary Slackness theorems for shift rules.
In Section~\ref{ssec:overview:ex-opt}, we discuss the existence of optimal finite feasible shift rules.
In Section~\ref{ssec:overview:the-list}, we give concrete results for sets~$\Xi$ that have been discussed in the literature.
We wrap things up in the concluding Section~\ref{sec:conclu}.

\section{Function spaces from expectation values of parameterized quantum circuits}\label{ssec:overview:fn-spaces}
Fix a $d\in\NN$, and let $\Xi$ be a fixed finite, symmetric subset of $\RR^d$.  We define the
vector space of functions
\begin{equation*}
  \Fun_\Xi
  :=
  \bigl\{   f\colon\RR^d\to\RR \bigm| f \text{ smooth and bounded } \Supp \hat f\subseteq\Xi \}
\end{equation*}
The Fourier transform of the bounded smooth function here is, technically, that of tempered distributions, but the finitness of
the spectrum along with the boundedness renders this machinery trivial: Indeed, we have
\begin{equation*}
  \Fun_\Xi
  =
  \bigl\{   f\colon x\mapsto \sum_{\xi\in\Xi} b_\xi e^{2\pii\xi\bullet x} \bigm| b\in\CC^\Xi \text{ with } b_{-\xi}=b_\xi^* \bigr\};
\end{equation*}
with $u\bullet v := \sum_j u_j v_j$.

\begin{quote}
  We emphasize our conventions for the Fourier transform:
  $\hat f(\xi) = \int e^{-2\pii\xi\bullet x} f(x)\, dx$ --- the crucial part being the $2\pi$ in
  the exponent!  This choice goes hand in hand with our decision to normalize Hamiltonian
  evolution to $t\mapsto e^{2\pii H\cdot t}$.  It leads to factors of powers of $2\pi$ in the
  derivatives and hence the standard deviations of the estimators of the derivatives.
\end{quote}

As indicated above, we find it adds clarity to the subject matter to think in terms of function
spaces rather than in terms of Hamiltonians or parameterized quantum circuits.
As this also adds some generality that is not required to address the currently known applications, we find it prudent to
restrict to function spaces which arise from practically relevant parameterized quantum circuits \textsl{whenever that simplifies
  the mathematics.}  However, we refer to the Proposition~\ref{prop:overview:FourierComputability} below
(\S\ref{ssec:overview:fn-spaces:FourierComputability}) as a justification for working with function spaces with arbitrary
frequency sets: Every function can all be realized as the expectation value function of a parameterized quantum circuit.

We now go through some common examples of parameterized quantum circuits and their function spaces.  What follows is well-known
and treated in detail elsewhere, e.g., \cite{Wierichs-Izaac-Wang-Lin:gen-shift:2021}.

\subsection{A single Hamiltonian}
The fact that evolution under a Hamiltonian produces a Fourier spectrum consisting of the
differences between the eigenvalues has been observed long ago, e.g.,
\cite{gilvidal-theis:inpu-redu:2019, GilVidal-Theis:CalcPQC:2018}.

For easy reference, we pack the statement into a proposition.  For convenience, let us define the \textit{difference set}
\begin{equation}\label{eq:def:Diff-op}
  \begin{split}
    \Diff\Lambda &:= \{\lambda' - \lambda \mid \lambda,\lambda'\in\Lambda\} \text{ for $\Lambda\subset\RR^d$;}
  \end{split}
\end{equation}
and the spectrum $\Spec(A)$, the set of eigenvalues of an operator~$A$.

\begin{proposition}\label{prop:overview:fou-spec}
  Consider a parameterized quantum circuit on~$n$ qubits of the following form.
  Let~$\varrho$ be a $n$-qubit density matrix, and~$\mu$ an $n$-qubit observable.  Let~$H$ be a
  Hermitian operator that acts on an arbitrary subset of the~$n$ qubits.  Consider the expectation
  value function
  \begin{equation}\label{def:exp-val-fun:1d}
    f\colon\RR\to\RR\colon
    x \mapsto \tr(\mu \; e^{-2\pi i H\cdot x}\, \varrho \, e^{2\pi i H\cdot x} ).
  \end{equation}
  The Fourier spectrum of~$f$ is contained in the set $\Diff\Spec(H)$.
\end{proposition}
\begin{proof}
  Abbreviating $\Lambda := \Spec(H)$ and replacing~$H$ by its spectral decomposition
  $H = \sum_{\lambda\in\Lambda} \lambda P_\lambda$, where the $P_*$ are the spectral projectors, we find, for every $x\in\RR$,
  \begin{multline*}
    \tr(\mu \; e^{-2\pi i H\cdot x}\, \varrho \, e^{2\pi i H\cdot x} )
    =
    \sum_{\lambda',\lambda\in\Lambda}
    \tr(\mu e^{-2\pii\lambda x} P_\lambda \varrho e^{2\pii\lambda' x} P_{\lambda'})
    \\
    =
    \sum_{\lambda',\lambda\in\Lambda}
    \tr(\mu P_\lambda \varrho P_{\lambda'}) \cdot e^{2\pii(\lambda'-\lambda)x}.
  \end{multline*}
\end{proof}

The function space associated with the case of Proposition~\ref{prop:overview:fou-spec} is $\Fun_{\Diff\Spec(H)}$, the space of
all real-valued functions whose Fourier spectum is contained in $\Diff\Spec(H)$.

\subsection{Several Hamiltonians}
In the multi-variate case, Hamiltonians $H_1,\dots,H_d$ can be paired with variables
$x_1,\dots,x_d$, respectively.  For example, for arbitrary quantum operations
$\mathcal{E}_1,\dots,\mathcal{E}_{d-1}$, and denoting by $[U]$ ($U$ a unitary) the Kraus-form
quantum operation $\varrho\mapsto U\varrho U^\dag$, we might have
\begin{multline}\label{def:exp-val-fun:any-d}
  f\colon\RR^d\to\RR\colon
  \\
  x \mapsto \tr(\mu \, \cdot \,
  \lt(
  \bigl[ e^{-2\pi i H_d\cdot x_d} \bigr]
  \circ
  \mathcal{E}_{d-1}
  \circ
  \bigl[ e^{-2\pi i H_{d-1}\cdot x_{d-1}} \bigr]
  \circ
  \dots
  \circ
  \mathcal{E}_{1}
  \circ
  \bigl[ e^{-2\pi i H_1\cdot x_1} \bigr]
  \rt)
  (\varrho) ).
\end{multline}
Applying Proposition~\ref{prop:overview:fou-spec} $d$ times yields, for the Fourier spectrum, the set
\begin{equation*}
  \Diff\Spec(H_1)\times\dots\times\Diff\Spec(H_d).
\end{equation*}
We can add a little more color by applying a linear transformation, i.e., ``$x=Au$'' for an arbitrary real $(d\times e)$-matrix
$A$.

\begin{proposition}\label{prop:overview:fou-spec-multi}
  If~$f$ is as in~\eqref{def:exp-val-fun:any-d}, and~$A$ is a real $(d\times e)$-matrix, then the function
  $g\colon u\mapsto f(Au)$ has Fourier spectrum contained in the set
  \begin{equation*}
    \{ A^\T\xi \mid \xi\in\Diff\Spec(H_1)\times\dots\times\Diff\Spec(H_d) \}.
  \end{equation*}
\end{proposition}
\begin{proof}[Sketch of proof.]
  Direct verification.  With $\Xi_f := \Diff\Spec(H_1)\times\dots\times\Diff\Spec(H_d)$, and
  $\Xi_g := \{A^\T\xi\mid\xi\in\Xi\}$, there are complex numbers $b_\xi$, $\xi\in\Xi$, such that
  \begin{equation*}
    g(u)
    = f(Au)
    = \sum_{\xi\in\Xi_f} b_\xi \, e^{2\pii \xi\bullet(Au)}
    = \sum_{\xi\in\Xi_f} b_\xi \, e^{2\pii (A^\T\xi)\bullet u}
    = \sum_{\eta\in \Xi_g} \biggl(\sum_\sstack{\xi\in\Xi\\A^\T\xi=\eta}  b_\xi \biggr)
    \; e^{2\pii \eta \bullet u}.
  \end{equation*}
\end{proof}

The function space associated with the case of Proposition~\ref{prop:overview:fou-spec-multi} is $\Fun_{\Xi_g}$ (re-using the
notation from the proof), the space of all real-valued functions each of whose Fourier frequencies is of the form $A^\T\xi$ for
some $\xi\in \Diff\Spec(H_1)\times\dots\times\Diff\Spec(H_d)$.

\subsection{Quantum circuits for arbitrary frequency sets}\label{ssec:overview:fn-spaces:FourierComputability}
We give the following quite obvious result as an argument in support of the generality in which the present paper is written. No
attempt has been made to minimize the resources.  The proof is in Appendix~\ref{apx:FourierComplexity}.

\begin{proposition}\label{prop:overview:FourierComputability}
  Given a $d\in\NN$, $\Xi\subset\RR^d$ finite, symmetric, and a $f\colon \RR^d\to\RR$ with Fourier spectrum equal to~$\Xi$,
  there exists a parameterized quantum circuit on at most $d\log_2(\abs{\Xi}+1)$ qubits and with~$d$ parameters, whose
  expectation value function equals~$f$.
\end{proposition}

Hence, we see that (a) the frequency sets of parameterized quantum circuits can be completely arbitrary, and (b) every function
in $\Fun_\Xi$ arises from a parameterized quantum circuit.

The practical value of Proposition~\ref{prop:overview:FourierComputability} is doubtful, as taking derivatives through shift
rules requires that the frequency sets are known.

\section{Finite-support shift rules, feasibility}\label{ssec:overview:def-feasible-shift-rule}
In this paper, a \textit{shift rule} is defined as a (finite complex Borel) measure~$\phi$ on $\RR^d$, which is applied to a
function $f\in\Fun_\Xi$ via convolution:
\begin{equation*}
  \phi * f(x) = \int f(x-u) \, d\phi(u)
\end{equation*}
We are interested in \textit{finite} shift rules, i.e., measures with finite support.  This restricts~$\phi$ to choosing a
finite set $A\subset\RR^d$ and complex numbers $u\in\CC^A$ (we will later restrict, wlog, to real coefficients $u\in\RR^A$), and
defining
\begin{equation}\label{eq:phi-shift-rules-ansatz}
  \phi_{A,u} := \sum_{a\in A} u_a \delta_a,
\end{equation}
where $\delta_a$ is the Dirac distribution / point measure at the point~$a$, which means $\int f(x) \,d\delta_a(x) = f(a)$ for
all reasonable functions~$f$.  For the whole convolution, we get
\begin{equation*}
  \phi_{A,u} * f\,(x) = \sum_{a\in A} u_a \, \delta_a * f = \sum_{a\in A} u_a f(x-a).
\end{equation*}
Given $\alpha\in\NN^d$, a shift rule~$\phi$ is called \textit{feasible} for $\Xi,\alpha$, if convolution against it realizes the
$\partial^\alpha$ partial derivative for all real-valued functions with Fourier spectrum contained in~$\Xi$, i.e.,
\begin{equation}\label{eq:def-basic-conv-deriv}
  \phi*f = \partial^\alpha f \text{ for all $f\in\Fun_\Xi$.}
\end{equation}

From now on, when we say ``shift rule'', we mean ``finite (but not necessarily feasible) shift rule''.

\section{System of equations and existence of (finite) feasible shift rules}\label{ssec:overview:sys-eqns}
Given a finite set $A\subset\RR^d$ and $u\in\CC^A$, applying the Fourier transform on both sides of
equation~\eqref{eq:def-basic-conv-deriv}, we find that $\phi$ is a feasible shift rule if and only if
\begin{equation}\label{primal:GlSys:hatphi-equals-2piixi}
  \hat\phi(\xi) = (2\pii\xi)^\alpha    \qquad\text{for all $\xi\in\Xi$;}
\end{equation}
we refer to the supplementary material for the details.
Using the fact that $\widehat{\delta_a}(\xi) = e^{-2\pii a\bullet\xi}$ for all $\xi\in\RR^d$, we can continue to an equivalent
system of equations formulated in the notation $\phi_{A,u}$ as in~\eqref{eq:phi-shift-rules-ansatz}:
\begin{equation}\label{primal:GlSys:sum-u_a_e^axi-equals-2piixi}
  \sum_{a\in A} u_a e^{ -2\pii a\bullet\xi } = (2\pi i \xi)^\alpha  \qquad\text{for all $\xi\in\Xi$.}
\end{equation}
As one can see, the system is linear in~$u$, with coefficients that are trigonometric functions of the vectors in the set~$A$.

To summarize: $\phi_{A,u}$ as in~\eqref{eq:phi-shift-rules-ansatz} is a feasible shift rule for $\Xi,\alpha$ if, and only if, it
is a solution to the the system of $\abs{\Xi}$ equations~\eqref{primal:GlSys:sum-u_a_e^axi-equals-2piixi}.

\mypar%
It is not, up to this point, obvious that for every wild combination of $\Xi$'s and~$\alpha$'s a feasible (finite!) shift rule
should even exist.  In the case where~$\Xi$ is equi-spaced the existence is easy (and known in the community, referred to as
``equidistant frequencies'' for $d=1$ in~\cite{Wierichs-Izaac-Wang-Lin:gen-shift:2021}).  The general case is discussed
in~\cite[Eqn.~(90)]{Wierichs-Izaac-Wang-Lin:gen-shift:2021}, but the author of the current paper wasn't able to find in Wierichs
et al.\ any attempt of an argument why their parameterized family of shift rules for the ``general case'' in $d=1$ should
contain one that's actually feasible (it's probably obvious and I'm just missing it).
We prove the existence of finite feasible shift rules for all $d,\Xi$ in the supplementary material.

We also note that, provided a single feasible shift rule exists, there exists a whole ``manifold''\footnote{Read: ``lots!''} of
them.

\section{Cost of a shift rule}\label{ssec:overview:cost}
There are several possibilities to define the ``cost'' of a shift rule.  The first one is
what~\cite{Wierichs-Izaac-Wang-Lin:gen-shift:2021} calls ``number of unique circuits'' and what in our terminology is the size
(cardinality) of the support of the measure: For $\phi_{A,u}$ as in~\eqref{eq:phi-shift-rules-ansatz}, the support is (contained
in) the set~$A$, and its size is~$\abs{A}$ (if all $u_*$ are non-zero).

In this paper the quantity we use as the cost of a shift rule $\phi$ is the usual norm\footnote{\label{def:norm-of-measure}%
  The one which underlies the total-variation: For a complex measure $\mu$ on $\Omega$, it is defined as
  $\displaystyle \Nm{\mu}_1 := \sup_{f} \int f\,d\mu$, where the supremum ranges over measurable functions with supremum norm at
  most~$1$. } %
of a finite measure, which we denote by $\Nm{\place}_1$.
Specializing to our situation, for $\phi_{A,u}$ as in~\eqref{eq:phi-shift-rules-ansatz}, we have
\begin{equation}\label{eq:def-cost}
  \cost(A,u) := \Nm{  \phi_{A,u}  }_1 = \sum_{a\in A} \abs{u_a}.
\end{equation}
As noted above, this definition is the same as Wierichs et al.'s ``number of shots'' up to factors.

Before we discuss the justification(s) for the choice of this cost function, we need to address complex vs.~real measures.  The
observant reader has noticed that we allow complex coefficients $u\in\CC^A$ in the shift
rules~\eqref{eq:phi-shift-rules-ansatz}, while in the literature, only real numbers have been considered.  It can be proven that
for every feasible shift rule~$\phi$, there always exists a feasible shift rule that is a real signed measure, i.e., has only
real coefficients when written in the form~\eqref{eq:phi-shift-rules-ansatz}), and whose cost is at most that of~$\phi$.

\mypar%
Hence, for the rest of the paper, we consider real coefficients $u\in\RR^A$ only.

\subsection{Cost of a shift rule as an error norm}  %
The first justification of the choice of the cost function comes from Harmonic/Numerical Analysis: It is a folklore
fact\footnote{%
  See~\cite{Rudin:real-complex-analysis:1987} for the background.  The inequality ``$\ge$'' follows from
  $\abs{\int f \,d\phi} \le \int \abs{f} \,d\abs{\phi} \le \Nm{f}_\infty \int d\abs{\phi} = \Nm{f}_\infty \Nm{\phi}_1$ where
  $\abs{\phi}$ is the total variation measure; the inequality ``$\le$'' is in the Riesz Representation Theorem.} %
that for every finite signed (i.e., real valued) Borel measures~$\phi$ on~$\RR^d$, we have
\begin{equation}\label{eq:1Nm-measure-equals-opnorm-integral}
  \Nm{\phi}_1
  =
  \sup
  \biggl\{   \absB{  \int f\,d\phi }
  \biggm|
  f \in \mathscr{C}_b(\RR^d)\text{, } \Nm{f}_\infty \le 1 \biggr\},
\end{equation}
where we have denoted by $\mathscr{C}_b$ the Banach space of continuous bounded functions, with
the supremum norm $\Nm{f}_\infty := \sup_{x} \abs{f(x)}$.

This implies that
\begin{equation*}
  \Nm{\phi}_1
  =
  \sup_{\Nm{f}_\infty\le 1} \Nm{ \phi*f }_\infty,
\end{equation*}
where the supremum extends over the functions in $\mathscr{C}_b(\RR^d)$ again.
In other words, $\Nm{\phi}_1$ is equal to the operator norm for the operator
\begin{equation*}
  D\colon \mathscr{C}_b(\RR^d) \to \mathscr{C}_b(\RR^d)
  \colon
  f \mapsto \phi*f.
\end{equation*}

As a general principle, norms of operators give error estimates.  Suppose for given $\Xi\subset\RR^d$, $\alpha$, we are
interested in the partial derivative $\partial^\alpha$ of a function $f\in\Fun_\Xi$, and we have a feasible finite-support shift rule
$\phi$ that accomplishes this.  But~$f$ is available only with an error (or noise), so the function that we can evaluate is
$\tilde f := f+h$ where $h\in\mathscr C_b(\RR^d)$.  In general, the Fourier frequencies of the error function, $h$, can be
outside of the set of frequencies~$\Xi$.  In the case of parameterized quantum circuits, where $\tilde f(x)$, for given~$x$, is
the average of discrete samples, $h$ can be expected to have arbitrarily large frequencies.

The question how much the error, $h$, messes up the computation of the derivative via convolution with~$\phi$ is upper bounded
by the operator norm.  In the worst case (i.e., supremum norm) over all $x\in\RR^d$, the noise is ``blown up'' by a factor of
$\Nm{\phi}_1$:
\begin{equation*}
  \Nm{  D\tilde f - Df  }_\infty
  =
  \Nm{ Dh }_\infty
  \le
  \Nm{D}\cdot\Nm{h}_\infty
  =
  \Nm{\phi}_1 \cdot \Nm{h}_\infty.
\end{equation*}
The inequality is tight, i.e., there exist $h$'s for which equality holds.

This reasoning is, of course, typical in any kind of Harmonic/Numerical Analysis context.

\subsection{Cost of a shift rule as worst-case standard deviation}
In this paragraph, we essentially repeat with a little more detail the calculation and approximation of Wierichs et al.\ that
leads to their definition of the ``number of shots'' in \cite[Eqn.~(14)]{Wierichs-Izaac-Wang-Lin:gen-shift:2021}.

Let us specialize to the case that the observable ($\mu$ in \eqref{def:exp-val-fun:1d}
and~\eqref{def:exp-val-fun:any-d}) has only $\pm1$ eigenvalues, e.g., is a (tensor product of)
Pauli obserable(s).  If~$f$ is such an expectation value function, then, for all $x\in\RR^d$, the
variance when evaluating $f(x)$ is $1-f^2(x)$.

As everybody knows, when the expectation of a finite sum $\sum_{a\in A} X_a$ of independent random variables $X_a$, $a\in A$, is
to be estimated by sampling each summand~$X_a$ independently a number of $N_a$ times and taking the average, then the
variance-minimizing way to distribute a budget~$N$ of shots, $N = \sum_a N_a$, is to choose $N_a$ proportional to the standard
deviation of $X_a$, for all~$a$.  As this is only possible in exceptional situations,\footnote{%
  For each~$a$, the standard deviation fraction $N\cdot \sigma_a / \sum_b\sigma_b$ would have to be an integer. } %
this perspective is \textsl{asymptotic,} i.e., we take $N\to\infty$.  In this limit, the resulting variance is asymptotic to
$\frac{1}{N} (\sum_a \sigma_a)^2$, where $\sigma_a$ is the standard deviation of $X_a$, for all~$a$.

In our case, when estimating a partial derivative at $x\in\RR^d$ by convolution with $\phi_{A,u}$ as
in~\eqref{eq:phi-shift-rules-ansatz} (with $u\in\RR^A$), we have $X_a = u_a\cdot F(x-a)$, $a\in A$, where for all
$(F(y))_{y\in\RR^d}$ is family of independent, $\pm1$-valued random variables with expectations $\Exp F(y) = f(y)$ for all~$y$.
The variance is $\abs{u_a}^2 (1-f(x-a)^2)$, and with optimally distributed shots, the asymptotic variance is
\begin{equation}
  \tfrac{1}{N} \lt( \sum_{a\in A} \abs{u_a}\sqrt{1-f^2(x-a)\,} \rt)^2.
\end{equation}
We take the square root, to obtain the asymptotic standard deviation
\begin{equation*}
  \tfrac{1}{\sqrt N}\sum_{a\in A} \abs{u_a}\sqrt{1-f^2(x-a)\,}.
\end{equation*}

To manage this expression, we are now taking the worst case over $x$ and~$f$: We are lower
bounding $f^2(x-a)$ by~$0$, for all $a\in A$.  This results in the upper bound of the standard
deviation of, asymptotically,
\begin{equation*}
  \tfrac{1}{\sqrt N} \sum_{a\in A} \abs{u_a} = \cost(A,u)/\sqrt{N}.
\end{equation*}

Taking the worst case is not just a convenient primitive of computational theory, in our case, it is also a bow to the dreaded
``Barren Plateaus'' phenomenon~\cite{McClean-Boixo-Smelyanskiy-Babbush-Neven:barren:2018}, where the expectation value functions
are flat.  (Clearly, lower bounding~$f^2$ by a constant(!) other than~$0$ doesn't change the comparison relations between
asymptotic standard deviations of different shift rules, and hence has no influence on which shift rule is optimal.)

We emphasize that we are very optimistic about the ingenuity of an estimator applying a shift rule given its support: We allow
it to somehow magically know how to distribute a budget of shots among the support of the shift rule in an optimal way,
depending on the values of the expectation value function at the support points!  Alas, in the worst case, it doesn't help at
all.

\subsection{Cost of a shift rule as a heuristic for sparsity minimization}\label{ssec:overview:cost:sparsity}
In the literature, the other cost factor for a shift rule is what~\cite{Wierichs-Izaac-Wang-Lin:gen-shift:2021} call ``number of
unique circuits'', and what we refer to as the size of the support of the Borel measure, the set~$A$ in the notation
of~\eqref{eq:phi-shift-rules-ansatz}.  Let us assume for a moment that our goal was not to minimize the standard deviation, but
to minimize the support, while a large but finite candidate set~$A$ that allows feasible shift rules with support contained in
it has been fixed.  The task at hand is to find a solution to a (finite, underdetermined) system of linear equations whose
support (number of non-zero entries of a vector) has minimal size:
\begin{equation*}
  \begin{split}
    {}                &\min \{a \mid u_a \ne 0 \}\\
    \text{subject to }&Eu = k
  \end{split}
\end{equation*}
for an appropriately chosen matrix~$E$ and vector~$k$; see above.  This is an NP-hard optimization problem, but, as we all know,
the standard heuristic for it is to replace the term under the minimization by the 1-norm: $\min \sum_a \abs{u_a}$:
\begin{equation}\label{eq:donoho}
  \begin{split}
    {}                &\min \sum_a \abs{u_a}\\
    \text{subject to }&Eu = k
  \end{split}
\end{equation}
In many cases, provable guarantees can be made on the effectiveness of that heuristic~\cite{Donoho-Tanner:sparsePNAS:2005}.

Hence, minimizing $\Nm{\phi}_1$ over all feasible shift rules can be justified completely independently from error norms and
worst-case standard deviations, simply as the goal to heuristically minimize the support of the shift rule, in other words,
maximize its sparsity.

\section{``Primal'', ``dual'' and weak duality}\label{ssec:overview:primal+dual+weakduality}
In optimization vernacular, the ``primal'' optimization problem is the one you want to solve, while the ``dual'' optimization
problem is one which helps you in solving or proving theorems about the primal.  The two are related by ``duality theorems''
which make statements about pairs of primal/dual feasible solutions.

From what we have discussed, given $d, \Xi, \alpha$, the optimal finite feasible shift rule is the
optimal solution (if attained) of the following optimization problem ($\phi$ is the variable):
\begin{subequations}\label{eq:overview:primal--beauty}
  \begin{align}
    \text{Infimum of }&\Nm{  \phi  }_1                                            \label{primal:obj}\\
    \text{subject to }&                                                           \notag\\
                      &\phi*f = \partial^\alpha f \text{ for all $f\in\Fun_\Xi$,} \label{primal:GlSys:phi*-equals-rond}\\
                      &\text{$\Supp\phi$  finite.}                                \label{primal:supp-phi--finite}
  \end{align}
\end{subequations}

We have noted above that the ostensibly infinite system of linear equations~\eqref{primal:GlSys:phi*-equals-rond} is equivalent
to each of the (finite) systems \eqref{primal:GlSys:hatphi-equals-2piixi} and~\eqref{primal:GlSys:sum-u_a_e^axi-equals-2piixi},
but for now, we stick with the form in~\eqref{eq:overview:primal--beauty}, because it gives an easy motivation for the Weak
Duality Theorem.

In terms of optimization theory, the finiteness condition is puzzling: If $d,\Xi,\alpha$ are such that the infimum is not
attained, do we have to pass to some kind of infinite-dimensional closure of our set of finite-support shift rules? We want to avoid
that, mainly because it isn't necessary, as we will see later.  The puzzlement is mitigated by considering the following
``dual'' optimization problem ($f$ is the variable):
\begin{subequations}\label{eq:overview:dual--beauty}
  \begin{align}
    \text{Supremum of }&(-\partial)^\alpha f(0)      \label{dual:obj}               \\
    \text{subject to  }&                             \notag                         \\
                       &\Nm{ f }_\infty \le 1        \label{dual:norm-f_infty-le-1} \\
                       &f\in \Fun_\Xi.               \label{dual:f-in-Fun_Xi}
  \end{align}
\end{subequations}
Note that $(-\partial)^\alpha=(-1)^{\abs{\alpha}}\partial^\alpha$.  (The notation $\abs{\alpha}$, for $\alpha\in\NN^d$, stands
for $\sum_j \alpha_j$.  )

The optimization problem~\eqref{eq:overview:dual--beauty} is of a form that is quite common in convex optimization, and both
theory and computation are manageable.
Indeed, the feasible region (meaning, the set of all \textit{dual feasible}~$f$, i.e., the~$f$'s satisfying the constraints
\eqref{dual:norm-f_infty-le-1},\eqref{dual:f-in-Fun_Xi}) is a subset of a finite dimensional real vector space (namely
$\Fun_\Xi$, of dimension $\abs{\Xi}$), and the subset is defined by infinitely many linear
constraints~\eqref{dual:norm-f_infty-le-1}:
\begin{equation}\label{dual:norm-f_infty-le-1:expanded}
  \text{for each $a\in\RR^d$:}\quad -1 \le f(a) \le 1.
\end{equation}
Moreover, the objective function $f\mapsto (-\partial)^\alpha f(0)$ is linear.

We have the following critical (albeit super easy) fact.
\begin{theorem}\label{thm:dual:nonempty-cpct-opt}
  Let $d\in\NN$, $\Xi\subset\RR^d$ finite symmetric, and $\alpha\in\NN^d$ be given.

  The dual feasible region is non-empty and compact.
  In particular, the objective function attains its supremum in at least one of its points.
\end{theorem}

For the proof of the theorem we refer to the supplementary material.

\subsection{Weak duality}\label{ssec:overview:primal+dual+weakduality:weakduality}
To give an idea why~\eqref{eq:overview:dual--beauty} can be considered a ``dual'' of~\eqref{eq:overview:dual--beauty}, we will
prove a ``weak duality theorem'':
\begin{quote}\slshape
  The objective function values of the feasible solutions to the maximization problem are less than or equal to the objective
  function values of the feasible solutions to the minimization problem.
\end{quote}

We give a slightly more general statement where a restriction to finite support of the measure is not necessary.

\begin{proposition}[Weak Duality Theorem]\label{prop:overview:weak-duality}
  Let~$\phi$ be a finite signed Borel measure on~$\RR^d$ satisfying~\eqref{primal:GlSys:phi*-equals-rond}, and let~$f$ be
  feasible for~\eqref{eq:overview:dual--beauty}.
  Then
  \begin{equation}
      (-\partial)^\alpha f(0)
      \le
      \Nm{\phi}_1.
  \end{equation}
\end{proposition}
\begin{proof}
  Abbreviating $x\mapsto f(-x)$ to $f(-\place)$, we calculate:
  \begin{equation}\label{eq:overview:ad-hoc-weak-duality}
    \begin{aligned}
      (-\partial)^\alpha f(0)
      &=
      \partial^\alpha(f(-\place))(0)
      \\
      &=
      \phi*(f(-\place))  (0)                &&\comment{by \eqref{primal:GlSys:phi*-equals-rond} and~\eqref{dual:f-in-Fun_Xi}}
      \\
      &=
      \int f(u-0) \, d\phi(u)
      \\
      &\le
      \Nm{f}_\infty \cdot \Nm{\phi}_1     &&\comment{by~\eqref{eq:1Nm-measure-equals-opnorm-integral}}
      \\
      &\le
      \Nm{\phi}_1.                        &&\comment{by~\eqref{dual:norm-f_infty-le-1}}
    \end{aligned}
  \end{equation}
  In the second equation, we have also used that for all~$f$, if $f \in\Fun_\Xi$ then $f(-\place) \in\Fun_\Xi$, as~$\Xi$ is
  symmetric.
\end{proof}

\section{The computational perspective}\label{ssec:overview:computation}
In this paper, we mostly use optimization as a tool for proving theorems, ultimately proving optimality of some feasible finite
shift rules.  The author finds that demonstrating how the optimization problems \eqref{eq:overview:primal--beauty}
and~\eqref{eq:overview:dual--beauty} could be solved computationally is not only fun, but also, importantly, helps understand
the theory.

This section is not mathematically rigorous.  For example, we will omit cases in the algorithms and gloss over details in the
justifications.  Making it rigorous would require to discuss too much background Convex Optimization theory.
In other words, while we will try to make as concrete and practical as possible \textsl{how} the methods work, we warn the
reader that understanding \textsl{why} they work is impossible without solid knowledge of Convex Optimization theory.

\subsection{Solving the dual optimization problem}\label{ssec:overview:computation:dual}
We start with the dual problem~\eqref{eq:overview:dual--beauty}, because its solution is canonical.  From a computational theory
perspective, the thrust of the matter is that, using the Ellipsoid algorithm, you can solve the optimization problem (to
provable optimality) efficiently, if you can (could!) solve efficiently the so-called ``separation problem'', which is an
algorithmic version of the separating hyperplane theorem: Given a point, either assert that it is feasible (i.e., inside the
convex set described by the constraints), or find an inequality that is valid for all feasible points and violated by the given
point.

We must supply a precision, $\eps>0$, to the Ellipsoid algorithm indicating that an ``$\eps$-almost feasible'' solution whose
objective function value is at most~$\eps$ worse than the mathematical optimum is sufficient; we refer
to~\cite{GrotschelLovaszSchrijver1993}, Definition~(2.1.10), for the exact formulation.  The running time will depend
polynomially on $\log(1/\eps)$.  Similarly, the separation problem should accept as input a precision $\delta>0$
(Definition~(2.1.13) in~\cite{GrotschelLovaszSchrijver1993}).

In this paper, where the computational aspect is illustrative and not essential for the theoretical part, we ignore the
computational complexity theory entirely, describing the practitioner's approach.
In our situation, from a practical point of view, the separation problem is the following:
\begin{quote}
  Given as input $d,\Xi$ and a function $f\in\Fun_\Xi$, return either the assertion $\Nm{f}_\infty \le 1$, or an $a\in\RR^d$
  such that $\abs{f(a)} > 1$ (where both inequalities are subject to a precision, $\delta>0$).
\end{quote}
The algorithm for solving~\eqref{eq:overview:dual--beauty} then consists in iteratively performing the following steps: Solve
the optimization problem with the infinitely many inequalities reduced to a finite list defined by a finite set $A\subset\RR^d$;
feed the resulting optimal (but not necessarily dual feasible) solution~$f^\star$ into an algorithm solving the separation
problem; update the finite set~$A$ of constraints and iterate, or exit with an optimal dual (near-)feasible solution.

Let's start by encoding functions in $\Fun_\Xi$.  Arbitrarily fix $\Xi_+ \subset \Xi$ in such a way that~$\Xi$ is the disjoint
union of the three sets $\{0\}$, $\Xi_+$, $-\Xi_+$.
Now, we represent the ``$f\in\Fun_\Xi$'' in~\eqref{eq:overview:dual--beauty} (and also in the input to the separation problem)
by its sine-cosine expansion:
\begin{equation}\label{LP:dual:optvariables}
  f
  =
  \mu
  + \sum_{\xi\in\Xi_+} \cos(2\pi \place \bullet\xi) \cdot x_\xi
  - \sin(2\pi \place\bullet\xi) \cdot y_\xi;
\end{equation}
the actual optimization variables are $\mu\in\RR$, $x,y \in\RR^{\Xi_+}$.  In that way the
constraint~\eqref{dual:f-in-Fun_Xi} is automatically satisfied.  The number of variables is, of course, equal to the dimension
of the real vector space $\Fun_\Xi$, namely $1+2\abs{\Xi_+} = \abs{\Xi}$.

The objective~\eqref{dual:obj} is obtained by taking $(-\partial)^\alpha$ at~$0$ of the function
in~\eqref{LP:dual:optvariables}.  With the convenient abbreviations
\begin{equation}\label{eq:def-shortcut-Re-Im}
  \begin{aligned}
    c_\xi &:= \Re((2\pii\xi)^\alpha), \quad\text{ and}\\
    s_\xi &:= \Im((2\pii\xi)^\alpha).
  \end{aligned}
\end{equation}
this gives\footnote{A small calculation is necessary to see this.  The reason why we set things up like this is because it goes
  nicely with duality.}
\begin{subequations}
  \begin{equation}\label{LP:dual:obj}
    \sum_{\xi\in\Xi_+}   c_\xi \cdot x_\xi  +  s_\xi \cdot y_\xi.
  \end{equation}

  The constraint~\eqref{dual:norm-f_infty-le-1} is written out in infinitely many linear inequalities as
  in~\eqref{dual:norm-f_infty-le-1:expanded}, two for every $a\in\RR^d$:
  \begin{equation}\label{LP:dual:two-ieqs}
    -1
    \le
    \quad
    \mu + \sum_{\xi\in\Xi_+} \cos(2\pi a\bullet\xi)\cdot x_\xi - \sin(2\pi a\bullet\xi)\cdot y_\xi
    \quad
    \le
    +1.
  \end{equation}
\end{subequations}

Fig.~\ref{fig:algo:dual} shows the complete algorithm for solving~\eqref{eq:overview:dual--beauty} to optimality --- provided we
\textsl{could} solve the separation problem.  We have added a tolerance parameter $\delta>0$ to replace the strict inequality in
the separation problem: Instead of the computationally unpractical $\abs{f(a)}>1$ we require $\abs{f(a)}>1+\delta$.  As a
consequence, the algorithm will, of course, return only a near-feasible solution.

\begin{figure}[ht]
  \begin{algorithm}[H]\caption{Solve-Dual}
    \Input{$d,\Xi_+\subset\RR^d,\alpha\in\NN^d$; $\delta\ge0$}
    \Output{Optimal near-feasible solution to~\eqref{eq:overview:dual--beauty} (in $\mu,x,y$ variables)}
    \BlankLine %
    \nl Initialize $A$ to some (finite) subset of $\RR^d$; e.g., $A:=\emptyset$.               \\
    \BlankLine                                                                                 %
    \nl                                                                                        %
    \Repeat{$A_+ = \emptyset$}{                                                                %
      \nl\label{algo:dual:line--LP}%
      \begin{minipage}[t]{0.9\linewidth}
        Solve the following Linear Program, in the variables $\mu\in\RR$, $x,y\in\RR^{\Xi_+}$:\\[1ex]
        $\displaystyle
        \begin{aligned}
          \qquad\qquad
          \textsl{maximize }   & \textsl{expression~\eqref{LP:dual:obj}}\\
          \textsl{subject to } & \textsl{inequalities~\eqref{LP:dual:two-ieqs}, for all $a\in A$}
        \end{aligned}
        $
      \end{minipage}                                                                             \\[1ex]
      \nl Obtain an optimal\footnote{%
        We are glossing over a detail here: The LP may be unbounded (indeed, if you start with $A=\emptyset$, that is the expected
        outcome), and the LP solver will return the data of a half-line contained in the LP feasible region.  Generally, separation
        algorithms can be adapted to handle that situation.  } %
      LP solution $\mu^\star, x^\star, y^\star$, with objective function value $v^\star$.\\
      \nl \SeparationRoutine{}{                                                     %
        \nl Initialize $A_+ := \emptyset$                                          \\
        \nl Find argmin $a_\text{min}$ and argmax $a_\text{max}$ of the function \\
        $\displaystyle\qquad\qquad
        f^\star\colon a\mapsto
        \mu^\star + \sum_{\xi\in\Xi_+} \cos(2\pi a\bullet\xi)\cdot x^\star_\xi - \sin(2\pi a\bullet\xi)\cdot y^\star_\xi$ \\
        \nl If $f^\star(a_\text{min}) < -1-\delta$, update $A_+ :=          \{a_\text{min}\}$  \\
        \nl If $f^\star(a_\text{max}) >  1+\delta$, update $A_+ := A_+ \cup \{a_\text{max}\}$  \\
        \nl \Return $A_+$                                                                      %
      }
      \BlankLine
      \nl Update $A := A\cup A_+$
    }
    \BlankLine                                                                                 %
    \nl \Return Optimal near-feasible solution $\mu^\star,x^\star,y^\star$ with objective function value~$v^\star$.
  \end{algorithm}
  \caption{Algorithm ``Solve-Dual''}\label{fig:algo:dual}
\end{figure}

Iterated LP algorithms with ``cutting plane generation'' (i.e., adding inequalities violated by the current optimal point),
often based on the Simplex algorithm for Linear Programming, are a standard tool in computational optimization to solve to
optimality convex optimization problems with linear objective where there is a large number of inequalities.

We invite the reader to try it out: The author has coded a toy demonstration, for $d=2$ only, of the algorithm as a notebook;
see Appendix~\ref{apx:source-code}.

\subsection{Finding an optimal shift rule}\label{ssec:overview:computation:primal}
A similar iterative approach, referred to as ``column generation'' in Linear Programming literature, would help us with the
primal problem~\eqref{eq:overview:primal--beauty}.  Since we have the finiteness-of-the-support
condition~\eqref{primal:supp-phi--finite}, a minimum may or may not be attained, but we can still hope for a shift rule that is
feasible and whose cost is at most~$\eps$ larger than the infimum, for a given $\eps>0$.  In that approach, in every iteration,
for a fixed finite set $A\subset\RR^d$, we would be solving the following Linear Program.  The variables are $u^+,u^- \in\RR^A$:
\begin{subequations}\label{eq:overview:primal-fixed-A-LP}
  \begin{align}
    \text{Minimum of } & \sum_{a\in A} u^+_a + u^-_a  \\
    \text{subject to } &
                         \label{eq:overview:primal-fixed-A-LP:eqns}
                         \begin{alignedat}[t]{4}
                           &                           &&
                           \sum_{a\in A} u^+_a - u^-_a                               &&= 0,    \\
                           &\forall\; \xi\in\Xi_+\colon&&
                           \sum_{a\in A} (u^+_a - u^-_a) \cos( 2\pi a \bullet\xi)    &&= c_\xi,\\
                           &\forall\; \xi\in\Xi_+\colon&&
                           \sum_{a\in A} (u^+_a - u^-_a) \sin(-2\pi a \bullet\xi)    &&= s_\xi,\\
                         \end{alignedat}\\
                       & u^+,u^- \ge 0
  \end{align}
\end{subequations}
The non-linear absolute value $\abs{u}$ that occurs in~\eqref{eq:overview:primal--beauty} via~\eqref{eq:def-cost} has been
``linearized'' by replacing the optimization variable~$u$ from~\eqref{eq:overview:primal--beauty} by its positive and negative
parts $u^\pm$, i.e., $u=u^+-u^-$.

Solving the LP will return either in the optimal feasible shift rule with support contained in~$A$, or the assertion that no
feasible shift rule with support contained in~$A$ exists.

We are not pursuing the column-generation approach here, mainly because it requires a lot of LP theory while not adding much
magic.  Indeed: The primal column generation method is mathematically exactly the same the dual column generation method: Each
step in one is the LP-dual of the corresponding step in the other.

\subparagraph{Finding an optimal shift rule when a dual-optimal basis is available.}  %
Instead, we now explain how we can computationally find a feasible finite-support shift rule whose cost equals that returned by
Algorithm~\ref{fig:algo:dual}, with $\eps,\delta$ set to~$0$ to simplify the exposition.  Indeed, the support of the resulting
shift rule will be contained in the last set~$A$ that was used to solve the LP in line~\ref{algo:dual:line--LP} there.  With the
notations used in Algorithm~\ref{fig:algo:dual}, take the values of the variables in the iteration where the Separation Routine
returns $A_+=\emptyset$, i.e., the last iteration.
We also assume that the optimal solution is a so-called ``vertex solution'' which comes with what is called a ``basis'' in LP
vernacular: A set of inequalities that are satisfied with equality by the vertex solution, and whose left-hand-side vectors are
linearly independent spanning a space of dimension~$\abs{\Xi}$.  (In other words: The coefficients of the selected inequalities
form a square, invertible matrix.)
We index the inequalities of the basis by the (disjoint) sets $A^+$ and $A^-$, where the elements $a\in A^+$ index inequalities
of the form ``$f(a) \le 1$'', and the elements of~$a\in A^-$ index inequalities of the form ``$-1 \le f(a)$''.

To obtain a feasible shift rule with support contained in~$A^+\cup A^- \subseteq A$ whose cost equals $v^\star$, we need to
solve for $u^+,u^-$, the system of linear equations consisting of equations~\eqref{eq:overview:primal-fixed-A-LP:eqns}, and
\begin{equation}\label{eq:overview:support-of-u^+,u^-}
  \begin{aligned}
    u^+_a &= 0 &\text{for all $a\notin A^+$} \\
    u^-_a &= 0 &\text{for all $a\notin A^-$.}
  \end{aligned}
\end{equation}
Linear Programming theory guarantees:
\begin{itemize}
\item The existence of a solution;
\item That the obtained $(u^+,u^-)$ is a feasible solution for the LP~\eqref{eq:overview:primal-fixed-A-LP}
\item That the cost of $u^+,u^-$ in the LP~\eqref{eq:overview:primal-fixed-A-LP} is equal to $v^\star$, and hence optimal by
  weak LP duality.
\end{itemize}

Set $u_a := u_a^+-u_a^-$ for all $a\in A$, \textit{et voil\`a,} $\phi_{A,u}$ defined as in~\eqref{eq:phi-shift-rules-ansatz} is
a the desired shift rule.

From the definition of an LP-basis above, we see that the size of the support $A^+\cup A^-$ of the shift rule is at most
$2\abs{\Xi_+}+1 = \abs{\Xi}$.

In our source code, see Appendix~\ref{apx:source-code}, we don't have to explicitly solve the system of linear equations: That
is implicitly done inside the LP solver: We can just query the result from the ``dual LP-solution''.

\subsection{Finding an optimal shift rule when a dual solution is known}\label{ssec:overview:computation:primal-w-knowledge}
We conclude this subsection by giving a method that, in the lattice-generating case, given an optimal dual solution ``on paper''
but no dual optimal basis, recovers an optimal shift rule by solving a single Linear Program.

Let $\Xi\subset\RR^d$ be (finite, symmetric, and) lattice generating, so that, wlog, we may assume that $\Xi\subset\ZZ^d$, so
that the functions in $\Fun_\Xi$ are $\ZZ^d$-periodic.

Denoting the known dual optimal solution by $f^\star$, let
\begin{equation*}
  A^\pm := \{ a \in \lt[-\nfrac12,\nfrac12\rt[^d \mid f^\star(a) = \pm 1\},
\end{equation*}
and (to simplify the notation) $A := A^+\cup A^-$.

Now, the LP is easily described: it is \eqref{eq:overview:primal-fixed-A-LP}, with the additional
equations~\eqref{eq:overview:support-of-u^+,u^-}.

The LP is of the form~\eqref{eq:donoho} from \S\ref{ssec:overview:cost:sparsity} in the motivation of the cost of a shift rule
based on a sparsity minimizing Linear Program.  This means that, if, for a given~$A$, if several optimal shift rules exist, we
can hope that one that is sparse will be returned.  That turns out to be true in practice: We have implemented a demonstration
of the method as a notebook, see Appendix~\ref{apx:source-code}.  The sparsest optimal shift rule is not always found.

It is important to understand that, because of ``complementary slackness'' (see below), every feasible solution of the LP is an
\textsl{optimal} shift rule, and that the objective function of the LP here does not contribute to optimality in terms of cost
of the shift rule, \textsl{only,} heuristically, to sparsity.

\mypar%
Now finally back to mathematical rigor.  Phew!

\section{Strong duality and complementary slackness}\label{ssec:overview:strong-duality}
Weak duality already has useful consequences: If you happen to stumble upon a feasible shift rule~$\phi$ along with a
dual-feasible function~$f$, such that their objective function values coincide, i.e., $(-\partial)^\alpha f(0) = \Nm{\phi}_1$
holds, then you can conclude that~$\phi$ is an \textsl{optimal} shift rule, and~$f$ is optimal for the dual.  If you are not
lucky enough to have that equation, you can still conclude from weak duality that the cost of your shift rule~$\phi$ is at most
$\Nm{\phi}_1 - (-\partial)^\alpha f(0)$ larger than the infimum.

\subsection{Strong duality}%
Strong duality adds to weak duality the certainty that if primal and dual optimal solutions exist, then their objective function
values coincide.  Our theorem is tailored for our situation.

\begin{theorem}[Strong duality of shift rule optimization]\label{thm:strong-duality}
  Let $d\in\NN$, $\Xi$ finite symmetric $\subset\RR^d$, and $\alpha\in\NN^d$ be given.

  \begin{enumerate}[label=(\alph*)]
  \item\label{thm:strong-duality:infimum} %
    The infimum in the primal optimization problem is equal to the maximum in the dual
    optimization problem.
  \item\label{thm:strong-duality:primalopt} %
    If~$\phi^\star$ is a finite feasible shift rule that is an optimal solution of the primal optimization problem, then there
    exists a dual feasible~$f$ such that
    \begin{equation}\label{thm:strong-duality:eqn}
      \Nm{\phi^\star}_1 = (-\partial)^\alpha f(0).
    \end{equation}
  \end{enumerate}
\end{theorem}

The way how Theorem~\ref{thm:strong-duality} is phrased might appear a bit odd once the complete picture is known: Optimal
finite-support shift rules always exist.  However, Theorem~\ref{thm:strong-duality} follows from Linear Programming theory, while for
the existence of optimal shift rules, we have to raise the stakes to graduate level math.

\subsection{Complementary slackness}%
``Complementary slackness'' is the Linear Programming theory equivalent of the Karush-Kuhn-Tucker (first-order) optimality conditions in general (convex)
optimization.  KKT is the inequality-constrained extension of Lagrange multipliers in multi-variable calculus.  In the
complementary slackness conditions, the dual solution takes the role of the multiplier.

\begin{theorem}[Complementary slackness for shift rule optimization]\label{thm:complementary-slackness}
  Let $d\in\NN$, $\Xi$ finite symmetric $\subset\RR^d$, and $\alpha\in\NN^d$ be given.

  If~$\phi$ is a finite feasible shift rule, and~$f$ dual feasible (i.e., feasible for~\eqref{eq:overview:dual--beauty}), then
  the following are equivalent:
  \begin{itemize}
  \item Both $\phi$ and~$f$ are optimal;
  \item $\Supp \phi \subseteq \{a\in\RR^d \mid \abs{f(a)} = 1 \}$ in such a way that for all~$a\in\RR^d$ the following logical
    implications hold:
    \begin{subequations}
      \begin{alignat}{3}
        \phi(a)&<0       &\quad\Longrightarrow&\quad&  f(a)&=-1,\\
        \intertext{and}
        \phi(a)&>0       &      \Longrightarrow    &&  f(a)&=+1.
      \end{alignat}
    \end{subequations}
  \end{itemize}
\end{theorem}

The proof of the theorem is in the supplementary material.

In applications of Linear Programming theory to other fields of theoretical computer science / mathematics, complementary
slackness is a powerful\footnote{Even if somewhat \dots vintage.} tool for proving theorems.  Not surprisingly,
Theorem~\ref{thm:complementary-slackness} is our main ingredient when proving optimality for some shift rules.

\section{Existence of optimal finite-support shift rules, and size of the support}\label{ssec:overview:ex-opt}
The attentive reader will have noticed that the question whether optimal (finite!) shift rules exist, is unto this point
unresolved.  Here is the answer.

\begin{theorem}\label{thm:overview:ex-opt}
  Let $d\in\NN$, let $\alpha\in\NN^d$, and $\Xi\subset\RR^d$ non-empty and symmetric.

  The optimization problem~\eqref{eq:overview:primal--beauty} attains an optimal solution, i.e., there exists a feasible
  (finite) shift rule~$\phi$ in which the infimum is attained.

  Moreover, there always exists an optimal shift rule whose support has size at most~$\abs{\Xi}$.
\end{theorem}

We also prove a strengthening of the size-of-the-support statement for $d=1$ to $\abs{\Xi}-(\alpha\%2)$ (where ``$\%$'' stands
for remainder upon division).
Wierichs et al.\ consider \textsl{only} shift-rules with that support, as they come out of their particular construction of
shift rules.  But they never ask the question whether shift rules with larger support but smaller cost (i.e., ``number of
shots'') might exist.

\begin{theorem}\label{thm:overview:ex-opt-sparse-1d}
  Let $\alpha\in\NN$, and $\Xi\subset\RR$ non-empty and symmetric.
  There always exists an optimal shift rule whose support has size at most~$\abs{\Xi} - (\alpha\%2)$.
\end{theorem}

\section{Optimal primal and/or dual solutions}\label{ssec:overview:the-list}
Before we approach the examples, we recall\footnote{This is standard undergrad math.} the formal definition of what we call
a ``lattice generating'' set of vectors, as promised in the introduction.

Consider the sub-group (wrt addition) of~$\RR^d$ generated by~$\Xi$.  There are two cases: This sub-group is a lattice, or this
sub-group is dense in a linear subspace of $\RR^d$.  A \textit{lattice} is a subgroup that is generated as a group by a linearly
(over the real numbers) independent set of vectors.  For~$\Xi$ to generate a lattice is equivalent to the matrix whose columns
are the elements of~$\Xi$ having rank over the rational numbers equal to (instead of strictly greater than) the rank over the
real numbers.  If $\Xi$ generates a lattice, we say, in this paper, that it is \textit{lattice generating.}

\subsection{Dimension $d=1$}\label{ssec:overview:the-list:d=1}
\begin{theorem}\label{thm:overview:1d-dual}
  Let $\alpha\in\NN$, and $\Xi\subset\RR$ non-empty and symmetric.  The optimal solution to the ``dual'' optimization
  problem~\eqref{eq:overview:dual--beauty} is
  \begin{equation*}
    a \mapsto (\partial^\alpha\cos)(2\pi\cdot\max\Xi\cdot a).
  \end{equation*}
  In particular, the cost of an optimal shift rule is $(2\pi\max\Xi)^\alpha$.
\end{theorem}

Note that the theorem holds without conditions on~$\Xi$ (beyond symmetry): It holds whether $\Xi$ is equi-spaced or not;
whether~$\Xi$ is lattice-generating or not; whether the space $\Fun_\Xi$ consists of periodic functions or not.
Indeed, we prove Theorem~\ref{thm:overview:1d-dual} in three steps, each of which is of interest in its own right.

\begin{enumerate}[label=\arabic*.]
\item\label{enum:overview:list-1d:equispaced}%
  For $\Xi=\{-K,\dots,K\}\subset\ZZ$ we establish primal and dual feasible solutions and prove complementary slackness.
  Complementary slackness relies on a fact about signs of derivatives of modified Dirichlet kernels; in this paper we give the
  proof of that lemma only for the practically relevant cases $\alpha=1,2,3,4$, and refer to a separate paper for the more
  technical general case.

  This proves that the shift rules in~\cite{Wierichs-Izaac-Wang-Lin:gen-shift:2021} for equi-spaced Fourier spectra are, indeed,
  optimal.

\item The case of lattice-generating Fourier spectra can be reduced to Step~\ref{enum:overview:list-1d:equispaced} by a cheap
  trick.

  As the proof establishes, in particular, the dual optimal value, and since an optimal finite-support shift rule exists
  (Theorem~\ref{thm:overview:ex-opt}), by strong duality (Theorem~\ref{thm:strong-duality}\ref{thm:strong-duality:primalopt}),
  we also know the cost of the best possible feasible shift rule --- which answers a question left open
  in~\cite{Wierichs-Izaac-Wang-Lin:gen-shift:2021}.

  Moreover, knowing the dual optimal solution allows us to realize the \textsl{computational} search for a sparse optimal shift
  rule through standard Linear Programming techniques: Solving a single LP as described
  in~\S\ref{ssec:overview:computation:primal-w-knowledge}.

\item Finally simultaneous Diophantine approximation can be used to reduce the non-lattice-generating case to the lattice
  generating one.
\end{enumerate}

\subsection{Dimension $d \ge 2$}\label{ssec:overview:the-list:d-ge-2}
We give optimal dual solutions only in a special case, for which we need some notation and terminology.  Our notation for the
canonical projection mappings is the following:
\begin{equation*}
  \pr_j\colon
  \RR^d \to \RR\colon
  x \mapsto x_j.
\end{equation*}
Let us say that a frequency set~$\Xi\subset\RR^d$ is \textit{pointy,} if
\begin{equation}\label{def:pointy-Xi}
  \text{there exists $\epsilon\in\{\pm1\}^d$ such that }
  \Bigl(
  \epsilon_1 \cdot \max\pr_1(\Xi),
  \epsilon_2 \cdot \max\pr_2(\Xi),
  \dots,
  \epsilon_d \cdot \max\pr_d(\Xi)
  \Bigr)^\T.
\end{equation}

\begin{theorem}\label{thm:overview:2d-dual}
  Let $d\ge 2$, $\alpha\in\NN^d$, and $\Xi\subset\RR^d$ non-empty and symmetric.

  If~$\Xi$ is pointy with $\epsilon$ as in~\eqref{def:pointy-Xi},
  then the optimal solution to the ``dual'' optimization problem~\eqref{eq:overview:dual--beauty} is
  \begin{equation*}
    a \mapsto \prod_{j=1}^d (\partial^{\alpha_j}\cos)(2\pi\cdot\max\pr_j(\Xi)\cdot a_j).
  \end{equation*}
  In particular, the cost of an optimal shift rule is $\prod_{j=1}^d (2\pi\max\pr_j(\Xi))^{\alpha_j}$.
\end{theorem}

We prove Theorem~\ref{thm:overview:2d-dual} in two steps.
\begin{enumerate}[label=\arabic*.]
\item\label{enum:overview:list-2d:product}%
  If $\Xi$ is a product set, i.e., if there exist symmetric sets $\Xi_j\subset\RR$, $j=1,\dots,d$, with
  $\Xi = \prod_{j=1}^d \Xi_j$, then we can reduce to~\eqref{thm:overview:1d-dual} together and
  Theorem~\ref{thm:overview:ex-opt}.

\item If $\Xi$ is pointy, then it is contained in a fitting product set, so the same cheap trick we used above
  gives a reduction to~\ref{enum:overview:list-2d:product}.
\end{enumerate}

\section{Conclusion}\label{sec:conclu}
This paper contributes what the author believes to be novel and useful tools to the study of efficient shift rules, culminating
in dual optimal solutions for a number of situations.  The tools inspire efficient computational methods for practically
computing optimal, hopefully sparse, shift rules for the many situations where (at least to date) no analytic description of the
optimal (sparse) shift rules are known.

Some questions remain, some of them ``merely'' mathematical, others with practical relevance.

\paragraph{Dual optimal solutions for $d>2$.} %
The present papers offers dual optimal solutions for $d>2$ only for a restricted family of frequency sets (the ``pointy'' ones).
Is it possible to write down and prove dual optimal solutions for all (lattice generating) frequency sets?  Is it true that (in
the lattice generating case), the dual optimal solution is always of the form
$a\mapsto \prod_{j=1}^d(\partial^{\alpha_j}\cos)(2\pi \xi_j\bullet a_j)$ for a $\xi$ in the lattice generated by~$\Xi$?  Or even
$\xi\in\Xi$?

\paragraph{Numerically stable implementations of the methods computing shift rules.} %
In the supplementary material (see Appendix~\ref{apx:source-code}) we have given quick-and-dirty implementations of the
computational methods in this paper, based on cheap Linear Programming solvers, and with no thought given to stability.
Obviously, a proper, numerically stable implementation of these methods is needed in order to apply them practically.  The most
interesting method is certainly that in \S\ref{ssec:overview:computation:primal-w-knowledge}, when the optimal dual solution is
known on paper as it allows to exploit all the theory of this paper.

\paragraph{Several partial derivatives at the same point}
Wierichs et al.\ give combined shift rules for estimating the whole Hessian.  Is their construction optimal?  The present paper
only shows how to estimate each individual entry of the Hessian optimally, but it does not consider estimators which ``recycle''
previously made measurements to estimate several partial derivatives at the same time.  (It can be shown that, for estimating
$\partial f(x)$, measurements at~$x$ cannot be used to reduce the cost, but that's a silly scenario.)

\section*{Acknowledgments}
This research was partly funded by the Estonian Centre of Excellence in Computer Science (EXCITE), and by Estonian Research
Agency (Eesti Teadusagentuur, ETIS) through grant PRG946.

\appendix
\section*{APPENDIX}
\section{Example source code}\label{apx:source-code}
The author of this paper has implemented some of the techniques discussed here in the mathematical computation programming
language Julia\footnote{\texttt{\href{https://julialang.org/}{julialang.org}}}, in the form of
Pluto\footnote{\texttt{\href{https://plutojl.org}{plutojl.org}}} notebooks.

The notebooks are available on
GitHub\footnote{\texttt{\href{https://github.com/dojt/ShiftRulesPluto}{github.com/dojt/ShiftRulesPluto}}} for download.  Here we
go through the files in that repository with a quick description of what the notebook does.

\paragraph{File \texttt{solve-for-u.jl}}
The main function in this notebook, \verb+solve_it(;+$\Xi_+$\verb+, +$\alpha$\verb+, +$A$\verb+)+, takes a set $\Xi_+\subset\QQ$
of (positive) frequencies, and a set $A\subset\RR$, and finds a shift rule of the form~\eqref{eq:phi-shift-rules-ansatz}.  If a
feasible shift rule with that support exists, it will be found; if none exists, one that is closest to being feasible is
returned.  Closest to being feasible refers to the 2-norm, in frequency space $\Nm{\partial^\alpha - \phi}_{2,\Xi}$.

The function uses floating point arithmetic, and the precision can be chosen between 64 bits and 16384 bits.  Changing the
precision can sometimes help recognizing a set~$A$ that yields no feasible shift rule, but where the error with 64-bit precision
is small, say $10^{-16}$.

\paragraph{File \texttt{dual-opt.jl}}
This notebook implements the method that is described in the subsection~\S\ref{ssec:overview:computation:dual}.  After
terminating with an optimal set~$A$, the optimal shift rule is obtained: This can be done directly from the dual variables,
which, is equivalent to the method described in~\S\ref{ssec:overview:computation:primal}, in the case when an optimal basis is
known.

\paragraph{File \texttt{sparse-optimal-1d.jl}}
In this notebook, knowledge about the optimal dual solution for lattice-generating frequency sets with $d=1$ is used to compute
an optimal feasible shift rule using a finite, i.e., ``normal'', Linear Program, as described
in~\S\ref{ssec:overview:computation:primal-w-knowledge}.

\section{Miscellaneous math background}\label{apx:miscmath}
\subsection{Lattices}\label{apx:miscmath:lattices-and-periodicity}
We sketch the proof of the equivalence of the periodicity of a function and the lattice-generating property of its Fourier
spectrum:

\begin{lemma}
  If $f$ is a tempered distribution on $\RR^d$ that is periodic (i.e., there exists an additive subgroup~$\Mu$ of $\RR^d$ such
  that $\delta_\mu*f = f$ for all $\mu\in\Mu$), then the support of the Fourier transform of~$f$ is a lattice.
\end{lemma}

If a function $f\colon\RR^d\to\RR$ has a Fourier spectrum~$\Xi$ that generates a lattice~$\Lambda$, then the function is
periodic modulo the dual lattice (which is, indeed, a lattice)
\begin{equation*}
  \Lambda^\bot := \{ \mu\in\RR^d \mid \forall \lambda\in\Lambda\colon \mu\bullet\lambda=0 \},
\end{equation*}
i.e., for all $x\in\RR^d$ and every $\mu\in\Lambda^\bot$ we have $f(x+\mu) = f(x)$.

For the other direction, we skip the part where you have to prove that the Fourier spectrum is discrete, and go right to the
part where we prove that it also generates a discrete set (i.e., a lattice).
If a smooth, bounded function $f\colon\RR^d\to\RR$ is periodic, meaning, there exists an additive subgroup $\Mu$
of~$\RR^d$ such that for all $x\in\RR^d$ and all $\mu\in\Mu$ we have $f(x+\mu)=f(x)$, then either $f=0$ or $\Mu$ is a lattice.
In the case when~$\Mu$ is a lattice, $\hat f = e^{2\pii \mu\bullet\place} \hat f$, which implies that, for all $\xi$ in the
support of $\hat f$, we have $\mu\bullet\xi \in \ZZ$, in other words $\Xi\subset \Mu^\bot$, so that the group generated by~$\Xi$
is also contained in $M^\bot$, and hence a lattice.

\subsection{Convexity}\label{apx:miscmath:cvx}
\subsubsection{Convex cones}\label{apx:miscmath:cvx:cvxcones}
For a set of vectors~$R$, let us denote by $\cvxcone(R)$ the convex cone generated by~$R$, i.e., the set of all (finite) linear
combinations with non-negative coefficients of elements of~$R$.  A convex cone~$C$ is called \textit{polyhedral,} if there
exists a finite set~$R$ such that $C=\cvxcone R$.

\subparagraph{Lemma on infinitely-generated convex cones.}  %
We use the following math notation only in this subsection.  For a set~$C$ in a finite-dimensional real vector space, we denote:
by $\bar C$ its topological closure; and by $C^\circ$ its topological interior.
Recall that, if~$C$ is a convex set (in a finite-dimensional real vector space) then $\bar C$ is also convex.

The following strengthened version of Carath\'eodordy's theorem is one of those things in convexity theory about which the
professor says ``It's obvious'' and then a student asks, ``Why?'', and the professor spends 60min proving it.  Carath\'eodordy's
theorem is the case $\bar C = C$.

\begin{lemma}\label{lem:miscmath:cvx:inf-gen}
  Let~$C$ be a convex cone in an $m$-dimensional real vector space, and suppose~$C$ is generated by a set of vectors~$R$.
  If $c \in (\bar C)^\circ$ then there exist $r_1,\dots,r_m\in R$ with $c \in \cvxcone\{r_1,\dots,r_m\}$.
\end{lemma}

\section{Complexity of function estimation}\label{apx:FourierComplexity}
In this section we briefly consider the following questions, ignoring ``subtleties'' about processing real numbers in classical
algorithms.

\begin{quote}\it
  Given a $d\in\NN$, $\Xi\subset\RR^d$ finite \& symmetric, and a $f\colon \RR^d\to\RR$ with Fourier spectrum equal to~$\Xi$:
  \begin{enumerate}[label=(\arabic*)]
  \item\label{apx:FourierComplexity:exist1}%
    Does there exist a parameterized quantum circuit (with~$d$ parameters) whose expectation value function is~$f$?
  \item\label{apx:FourierComplexity:existUniform}%
    Does there exist a uniform family of these circuits, i.e., a polynomial time classical algorithm which, upon input of
    $d,\Xi,\hat f$, produces the PQC as in~\ref{apx:FourierComplexity:exist1}?
  \item How does the minmum depth (or other complexity measure) of the quantum circuits in~\ref{apx:FourierComplexity:exist1}
    relate to the complexity of estimating the function with a (randomized) classical circuit/algorithm?
  \end{enumerate}
\end{quote}

The questions are informally inspired by the complexity theory of Boolean functions.
We give the obvious, cheap answer to~\ref{apx:FourierComplexity:exist1}, Proposition~\ref{prop:overview:FourierComputability},
as an argument in support of the generality in which the present paper has been written, making absolutely no attempt to search
for an efficient parameterized quantum circuit --- and express the hope that the other questions might be picked up by
\textsl{smart} people.

Recall the definition of $\Diff$ from~\eqref{eq:def:Diff-op}, and, for conciseness, denote by $\Spec(A)$ the set of eigenvalues
of a given operator~$A$.

\begin{proof}[Proof of Proposition~\ref{prop:overview:FourierComputability}]
  For each $j=1,\dots,d$, consider the set $\Xi_j = \{\xi_j \mid \xi\in\Xi\}$ --- the projection of~$\Xi$ onto the $j$th
  coordinate.  For each $j=1,\dots,d$, apply Lemma~\ref{lem:apx:FourierComplexity:d=1-H} below to each of the~$\Xi_j$, and
  define $q_j := \ceil{ \log_2( (\abs{\Xi_j}+1)/2 ) }$.  Letting
  \begin{equation*}
    q := \sum_{j=1}^d  q_j \quad < d\, \log_2(\abs{\Xi}+1),
  \end{equation*}
  we take~$q$ qubits, split into~$d$ quantum registers, where, for $j=1,\dots,d$, register number~$j$ has~$q_j$ qubits.  Let us
  denote by~$H_j$ the operator from Lemma~\ref{lem:apx:FourierComplexity:d=1-H} below applied to $\Xi_j$, acting on the $j$th
  quantum register.  We then let
  \begin{equation*}
    H := \sum_{j=1}^d H_j,
  \end{equation*}
  and consider the parameterized circuit with initial state $\ket{\psi_0}$, multi-parameterized unitary evolution $U(x)$, as
  follows:
  \begin{equation*}
    \begin{split}
      \ket{\psi_0} &:= 2^{-q/2} \sum_{\gamma\in\{0,1\}^q} \ket{\gamma} \\
      U(x)         &:= e^{2\pii \sum_{j=1}^d x_j H_j}.
    \end{split}
  \end{equation*}
  Now we define a Hermitian $q$-qubit operator~$\mu$ as follows.

  First of all, we interpret each of the~$d$ quantum registers in the computational basis as encoding an unsigned integer.
  For $\ell_j,\ell'_j\in\{0,\dots,2^{q_j}\}$, $j=1,\dots,d$, we say that $(\ell,\ell')$ \textit{fits,} if, for all
  $j=1,\dots,d$, we have both $\ell'_j,\ell_j\le (\abs{\Xi_j}-1)/2$ and $\ell_j\cdot\ell'_j=0$.
  Now, for $\ell_j,\ell'_j\in\{0,\dots,2^{q_j}\}$, we define
  \begin{equation*}
    \bra{\ell'_1,\ell'_2,\dots,\ell'_d} \mu \ket{\ell_1,\ell_2,\dots,\ell_d}
    :=
    \begin{cases}
      2^{q/2} \hat f(\xi),  & \text{if $(\ell,\ell')$ fits and }
      {}                      \Bigl( \xi_{\ell'_1}-\xi_{\ell_1}\,,\; \dots\,,\; \xi_{\ell'_d}-\xi_{\ell_d} \Bigr)^\T \in \Xi  \\
      0,                    & \text{otherwise.}
    \end{cases}
  \end{equation*}
  Note that, as~$f$ is real valued, $\mu$ is indeed Hermitian.

  The final quantum circuit is the following: Starting from $\ket0$, prepare a uniform superposition over all computational
  basis states, then apply the parameterized unitary (using the usual tools, as the exponent is a linear combination of
  tensor-products of $Z$-operators), finally measure~$\mu$.  The last step is done by diagonalizing~$\mu$ using a basis change
  unitary~$W$ from the computational basis, i.e., $W\mu W^\dagger$ is diagonal in the computational basis, and implementing~$W$
  as a quantum circuit, based on the usual universality theorems (e.g., Section~4.5.2 in \cite{Nielsen-Chuang:book:2000}).
\end{proof}

\begin{lemma}[Trivial Helper Lemma]\label{lem:apx:FourierComplexity:d=1-H}
  Let $\Xi\in\RR$ be finite and symmetric, and let $0 =: \xi^{(0)} < \xi^{(1)} < \dots < \xi^{(m)}$, where
  $m := (\abs{\Xi}-1)/2$, be the non-negative elements of~$\Xi$.  There exists a Hermitian $\ceil{ \log_2(m+1) }$-qubit
  operator~$H$ such that $\Diff\Spec(H) \supseteq \Xi$.

  More precisely, the operator~$H$ is diagonal in the multi-qubit Z-basis (computational basis), and for
  $\ell,\ell' \in \{0,\dots,m\}$ (coded in binary):
  \begin{equation}\label{def:apx:FourierComplexity:d=1-H}
    \bra{\ell'} H \ket{\ell} = \delta_{\ell',\ell} \cdot \xi^{(\ell)}.
  \end{equation}
\end{lemma}

Let $\gamma\in\NN^q$, and a 1-qubit operator~$A$, we use the shorthand
\begin{equation*}
  A^{\otimes\gamma} := A^{\gamma_1}\otimes A^{\gamma_2}\otimes \dots \otimes A^{\gamma_{q-1}}\otimes A^{\gamma_q}.
\end{equation*}

\begin{proof}[Proof of Lemma~\ref{lem:apx:FourierComplexity:d=1-H}]
  To define~$H$, follow~\eqref{def:apx:FourierComplexity:d=1-H}, and choose the remaining $2^q-m-1$ diagonal entries
  arbitrarily.  We have $\Diff\Spec(H) \supseteq \Diff\{\xi^{(0)},\dots,\xi^{(m)}\} \supseteq \Xi$.
\end{proof}


\bibliographystyle{plain}

\begin{thebibliography}{10}

\bibitem{Banchi-Crooks:stoch-param-shift:2021}
Leonardo Banchi and Gavin~E. Crooks.
\newblock Measuring analytic gradients of general quantum evolution with the
  stochastic parameter shift rule.
\newblock {\em {Quantum}}, 5:386, January 2021.

\bibitem{Benedetti-Lloyd-Sack-Fiorentini:PQCs:2019}
Marcello Benedetti, Erika Lloyd, Stefan Sack, and Mattia Fiorentini.
\newblock Parameterized quantum circuits as machine learning models.
\newblock {\em Quantum Science and Technology}, 4(4):043001, 2019.

\bibitem{Cerezo-Arrasmith-Babbush-Benjamin-etal:VQAs:2021}
Marco Cerezo, Andrew Arrasmith, Ryan Babbush, Simon~C Benjamin, Suguru Endo,
  Keisuke Fujii, Jarrod~R McClean, Kosuke Mitarai, Xiao Yuan, Lukasz Cincio,
  et~al.
\newblock Variational quantum algorithms.
\newblock {\em Nature Reviews Physics}, 3(9):625--644, 2021.

\bibitem{Crooks:param-shift:2019}
Gavin~E. Crooks.
\newblock Gradients of parameterized quantum gates using the parameter-shift
  rule and gate decomposition.
\newblock {\em Preprint}, 2019.
\newblock arXiv:1905.13311.

\bibitem{Donoho-Tanner:sparsePNAS:2005}
David~L Donoho and Jared Tanner.
\newblock Sparse nonnegative solution of underdetermined linear equations by
  linear programming.
\newblock {\em Proceedings of the National Academy of Sciences},
  102(27):9446--9451, 2005.

\bibitem{Farhi-Goldstone-Gutmann:QAOA:2014}
Edward Farhi, Jeffrey Goldstone, and Sam Gutmann.
\newblock A quantum approximate optimization algorithm.
\newblock {\em arXiv preprint arXiv:1411.4028}, 2014.

\bibitem{gilvidal-theis:inpu-redu:2019}
Javier Gil~Vidal and Dirk~Oliver Theis.
\newblock Input redundancy for parameterized quantum circuits.
\newblock {\em Preprint}, 2019.
\newblock arXiv:1901.11434.

\bibitem{GrotschelLovaszSchrijver1993}
M.~Gr{\"o}tschel, L.~Lov{\'a}sz, and A.~Schrijver.
\newblock {\em Geometric algorithms and combinatorial optimization}, volume~2
  of {\em Algorithms and {C}ombinatorics}.
\newblock Springer, 2nd edition, 1993.

\bibitem{Hubregtsen-Wilde-Qasim-Eisert:grad-rules:2021}
Thomas Hubregtsen, Frederik Wilde, Shozab Qasim, and Jens Eisert.
\newblock Single-component gradient rules for variational quantum algorithms.
\newblock {\em Preprint}, 2021.
\newblock arXiv:2106.01388v1.

\bibitem{Izmaylov-Lang-Yen:algebraic-param-shift:2021}
Artur~F. Izmaylov, Robert~A. Lang, and Tzu-Ching Yen.
\newblock Analytic gradients in variational quantum algorithms: Algebraic
  extensions of the parameter-shift rule to general unitary transformations.
\newblock {\em Preprint}, 2021.
\newblock arXiv:2107.08131.

\bibitem{Kottmann-Anand-AspuruGuzik:diffUCC:2021}
Jakob~S. Kottmann, Abhinav Anand, and Al{\'a}n Aspuru-Guzik.
\newblock A feasible approach for automatically differentiable unitary
  coupled-cluster on quantum computers.
\newblock {\em Chemical Science}, 12(10):3497--3508, 2021.

\bibitem{Kyriienko-Elfving:genqcdiff:2021}
Oleksandr Kyriienko and Vincent~E. Elfving.
\newblock Generalized quantum circuit differentiation rules.
\newblock {\em Preprint}, 2021.
\newblock arXiv:arXiv:2108.01218.

\bibitem{Mari-Bromley-Killoran:higher-order:2021}
Andrea Mari, Thomas~R. Bromley, and Nathan Killoran.
\newblock Estimating the gradient and higher-order derivatives on quantum
  hardware.
\newblock {\em Phys. Rev. A}, 103:012405, Jan 2021.

\bibitem{McClean-Boixo-Smelyanskiy-Babbush-Neven:barren:2018}
Jarrod~R McClean, Sergio Boixo, Vadim~N Smelyanskiy, Ryan Babbush, and Hartmut
  Neven.
\newblock Barren plateaus in quantum neural network training landscapes.
\newblock {\em Nature communications}, 9(1):4812, 2018.

\bibitem{Mitarai-Negoro-Kitagawa-Fujii:qcirclearn:2018}
Kosuke Mitarai, Makoto Negoro, Masahiro Kitagawa, and Keisuke Fujii.
\newblock Quantum circuit learning.
\newblock {\em Phys. Rev. A}, 98:032309, 2018.
\newblock arXiv:1803.00745.

\bibitem{Nielsen-Chuang:book:2000}
Michael~A. Nielsen and Isaac~L. Chuang.
\newblock {\em Quantum Computation and Quantum Information}.
\newblock Cambridge University Press, 2000.

\bibitem{Ostaszewski-Grant-Benedetti:rotosolve:2021}
Mateusz Ostaszewski, Edward Grant, and Marcello Benedetti.
\newblock Structure optimization for parameterized quantum circuits.
\newblock {\em {Quantum}}, 5:391, January 2021.

\bibitem{Schuld-Bergholm-Gogolin-Izaac-Killoran:eval-anal-grad:2019}
Maria Schuld, Ville Bergholm, Christian Gogolin, Josh Izaac, and Nathan
  Killoran.
\newblock Evaluating analytic gradients on quantum hardware.
\newblock {\em Phys. Rev. A}, 99:032331, Mar 2019.

\bibitem{Stokes-Izaac-Killoran-Carleo:qnatgrad:2020}
James Stokes, Josh Izaac, Nathan Killoran, and Giuseppe Carleo.
\newblock Quantum natural gradient.
\newblock {\em Quantum}, 4:269, May 2020.

\bibitem{GilVidal-Theis:CalcPQC:2018}
Javier~Gil Vidal and Dirk~Oliver Theis.
\newblock Calculus on parameterized quantum circuits.
\newblock {\em Preprint arXiv:1812.06323}, 2018.

\bibitem{Rudin:real-complex-analysis:1987}
Rudin Walter.
\newblock {\em Real and Complex Analysis}.
\newblock McGraw-Hill, 1987.

\bibitem{Wierichs-Izaac-Wang-Lin:gen-shift:2021}
David Wierichs, Josh Izaac, Cody Wang, and Cedric Yen-Yu Lin.
\newblock General parameter-shift rules for quantum gradients.
\newblock {\em Preprint}, 2021.
\newblock arXiv:2107.12390.

\bibitem{Yuan-Suguru-Qi-Ying-Benjamin:theory-VQS:2019}
Xiao Yuan, Suguru Endo, Qi~Zhao, Ying Li, and Simon~C Benjamin.
\newblock Theory of variational quantum simulation.
\newblock {\em Quantum}, 3:191, 2019.

\end{thebibliography}

\end{document}